\documentclass[11pt]{article}%
\usepackage{fullpage}
\usepackage{amssymb,amsmath}
\usepackage{graphicx, epsfig}
\usepackage{float}

\usepackage[compact]{titlesec}

\def\idrm#1{\ensuremath{\mathrm{#1}}}

\def\floor#1{\lfloor #1 \rfloor}
\def\ceil#1{\lceil #1 \rceil}

\newcommand{\no}[1]{}

\newcommand{\todo}[1]{} 

\newtheorem{theorem}{Theorem}
\newtheorem{lemma}{Lemma}

\newtheorem{proposition}{Proposition}

\newenvironment{proof}{\trivlist\item[]\emph{Proof}:}%
{\unskip\nobreak\hskip 1em plus 1fil\nobreak$\Box$
\parfillskip=0pt%
\endtrivlist}

\newenvironment{itemize*}%
  {\begin{itemize}%
    \setlength{\itemsep}{0pt}%
    \setlength{\parskip}{0pt}%
    \setlength{\parsep}{0pt}%
    \setlength{\topsep}{0pt}%
    \setlength{\partopsep}{0pt}%
  }%
  {\end{itemize}}%



\newcommand{\cS}{{\cal S}}
\newcommand{\oT}{{\overline T}}
\newcommand{\ocT}{{\overline {\cal T}}}
\newcommand{\op}{{\overline p}}

\newcommand{\cT}{{\cal T}}

\newcommand{\oB}{{\overline B}}

\newcommand{\oP}{{\overline P}}

\newcommand{\eps}{\varepsilon}

\newcommand{\occ}{\mathrm{occ}}

\newcommand{\extendright}{\mathtt{extendright}}
\newcommand{\contractleft}{\mathtt{contractleft}}
\newcommand{\wlink}{\mathtt{wlink}}

\newcommand{\ra}{\idrm{rank}}
\newcommand{\sel}{\idrm{select}}
\newcommand{\acc}{\idrm{access}}
\newcommand{\lab}{\idrm{lab}}
\newcommand{\Acc}{\mathit{Acc}}
\newcommand{\longver}[1]{}

\begin{document}
\title{Space-Efficient Construction of Compressed Indexes in Deterministic Linear Time}
\author{
 J. Ian Munro\thanks{Cheriton School of Computer Science, University of Waterloo. Email {\tt imunro@uwaterloo.ca}.}
 \and
 Gonzalo Navarro\thanks{CeBiB --- Center of Biotechnology and Bioengineering, Department of Computer Science, University of Chile. Email {\tt gnavarro@dcc.uchile.cl}. Funded with Basal Funds FB0001, Conicyt, Chile.}
\and 
 Yakov Nekrich\thanks{Cheriton School of Computer Science, University of Waterloo.
 Email: {\tt yakov.nekrich@googlemail.com}.}
 }
\date{}
\maketitle

\thispagestyle{empty}
\begin{abstract}
We show that the compressed suffix array and the compressed suffix tree  of a string $T$ can be built in $O(n)$ deterministic time using $O(n\log\sigma)$  bits of space, where $n$ is the string length and $\sigma$ is the alphabet size.  Previously described  deterministic  algorithms either run in  time that depends on the alphabet size or need $\omega(n\log \sigma)$ bits of working space. Our result has immediate applications to other problems, such as yielding the first deterministic linear-time LZ77 and LZ78 parsing algorithms that use $O(n \log\sigma)$ bits.
\end{abstract}
\newpage
\setcounter{page}{1}

\section{Introduction}
\label{sec:intro}
In the string indexing problem we pre-process a string $T$, so that for any query string $P$ all occurrences of $P$ in $T$ can be found efficiently. Suffix trees and suffix arrays are two most popular solutions of this fundamental problem. A suffix tree is a compressed trie on  suffixes of $T$; it enables us to find 
all occurrences of a string $P$ in $T$ in time $O(|P|+\occ)$ where $\occ$ is the number of times $P$ occurs in $T$ and $|P|$ denotes the length of $P$. In addition to indexing, suffix trees also support a number of other, more sophisticated, queries.  
The suffix array of a string $T$ is the lexicographically sorted array of its suffixes. Although suffix arrays  do not support all queries that can be answered by the suffix tree, they  use less space and are more popular in practical implementations. While the suffix tree occupies $O(n\log n)$ bits of space, the suffix array can be stored in $n\log n$ bits.

During the last twenty years there has been a significant increase in interest in compressed indexes, i.e., data structures that keep $T$ in compressed form and support string matching queries. The compressed suffix array (CSA)~\cite{GrossiV05,FerraginaM05,Sad03} and the compressed suffix tree (CST)~\cite{Sadakane07} are compressed counterparts  of the suffix array and the suffix tree respectively.  A significant part of compressed indexes relies on these two data structures or their variants.   Both CSA and CST can be stored in $O(n\log \sigma)$ bits or less; we refer to e.g.~\cite{BelazzouguiN14} or~\cite{NM06} for an overview of compressed indexes.

It is well known that both the suffix array and the suffix tree can be constructed in $O(n)$ time~\cite{McCreight76,Ukkonen95,Weiner73, KSPP05}. The first algorithm 
that constructs the suffix tree in linear time independently of the alphabet size was presented by Farach~\cite{Farach97}.  There are also algorithms that directly construct the suffix array of $T$ in $O(n)$ time~\cite{KarkkainenSB06,KoA05}.
If the (uncompressed) suffix tree is available, we can obtain 
CST and CSA in $O(n)$ time. However this approach requires $O(n\log n)$ bits of working space. 
The situation is different if we want to construct compressed variants of these data structures using only $O(n\log \sigma)$ bits of space.  Within this space the algorithm of Hon et al.~\cite{HonSS09} constructs the CST in $O(n\log^{\eps}n)$ time for an arbitrarily small constant $\eps>0$. In the same paper the authors also showed that CSA can be constructed in $O(n\log\log \sigma)$ time. The algorithm of Okanohara and Sadakane constructs the CSA in linear time, but needs $O(n\log\sigma\log\log n)$ bits of space~\cite{OkanoharaS09}.  
Belazzougui~\cite{Belaz14} described randomized algorithms that build both CSA and CST in $O(n)$ time and $O(n\log\sigma)$ bits of space. His approach also provides deterministic algorithms with runtime $O(n\log\log \sigma)$~\cite{Belaz14arx}.   In this paper we show that randomization is not necessary in order to construct CSA and CST in linear time. Our algorithms run in $O(n)$ deterministic time and require $O(n\log\sigma)$ bits of space. 

Suffix trees, in addition to being an important part of many compressed indexes, also play an important role in many string algorithms.  One prominent example is Lempel-Ziv parsing of a string using $O(n\log\sigma)$ bits. The best previous solutions for this problem either take $O(n\log\log\sigma)$ deterministic time or $O(n)$ randomized time~\cite{KS16,BelazzouguiP16}. For instance K{\"o}ppl and Sadakane~\cite{KS16} showed how we can obtain LZ77- and LZ78-parsing for a string $T$ in $O(n)$ deterministic time and $O(n\log\sigma)$ bits, provided that the CST of $T$ is constructed. Thus our algorithm, combined with  their results, leads to the first  linear-time deterministic  LZ-parsing  algorithm that needs $O(n\log\sigma)$ bits of space. 

\paragraph{Overview.} 
The main idea of our approach is the use of batch processing.  Certain operations, such as rank and select queries on sequences, are a bottleneck of previous deterministic solutions. Our algorithms are divided into a large number of small tasks that can be executed independently.  Hence, we can collect large batches of queries and answer all queries in a batch. This approach speeds up the computation because, as will be shown later,  answering all queries in a batch takes less time than answering the same set of queries one-by-one. For example, our algorithm for generating the Burrows-Wheeler Transform of a text $T$ works as follows.  We cut the original text into slices of $\Delta=\log_{\sigma}n$ symbols.  The BWT sequence is constructed by scanning all slices in the right-to-left order. All slices are processed at the same time. That is, the algorithm  works in $\Delta$ steps and during the $j$-th step, for  $0\le j \le \Delta-1$,  we process all suffixes that start at position $i\Delta-j-1$  for all  $1\le i\le n/\Delta$.  Our algorithm maintains the sorted list of suffixes and keeps information about those suffixes in a symbol sequence $B$.   For every suffix $S_i=T[i\Delta -j-1 ..]$ processed during the step $j$, we must find its position in the sorted list of suffixes. Then the symbol $T[i\Delta-j-2]$ is inserted at the position that corresponds to $S_i$ in $B$.  Essentially we can find the position  of every new suffix $S_i$ by answering a rank query on the sequence $B$. Details are given in Section~\ref{sec:lintimebwt}.  Next  we must update the sequence  by inserting the new symbols into $B$.  Unfortunately we need $\Omega(\log n/\log \log n)$ time in general to answer rank queries on a dynamic sequence ~\cite{FS89}. Even if we do not have to update the sequence, we need $\Omega(\log\log \sigma)$ time to answer a rank query~\cite{BelazzouguiN15}.  In our case, however, the scenario is different: There is no need to answer queries one-by-one. We must provide answers to a large \emph{batch} of $n/\Delta$ rank queries with one procedure.  In this paper we show that the lower bounds for rank queries can be circumvented in the batched scenario: we can answer the  batch of  queries in $O(n/\Delta)$ time, i.e., in constant time per query. We also demonstrate that a batch of $n/\Delta$ insertions can be processed in $O(n/\Delta)$ time. This result is of independent interest. 
 
Data structures that answer batches of rank queries and support batched updates are  described in Sections~\ref{sec:batchrank},~\ref{sec:listlabel}, and~\ref{sec:batchdynseq}. This is the most technically involved aspect of our result.  In Section~\ref{sec:batchrank} we show how answers to a large batch  of queries can be provided. In Section~\ref{sec:listlabel} we describe a special labeling scheme that assigns monotonously increasing labels to elements of a list. We conclude this portion in Section~\ref{sec:batchdynseq} where we show how the static data structure can be dynamized.  Next we turn to the problem of constructing the compressed suffix tree.  First we describe a data structure that answers partial rank queries in 
constant time and uses $O(n\log\log \sigma)$ additional bits in Section~\ref{sec:partrank}; unlike previous solutions, our data structure can be constructed in $O(n)$ deterministic time.  This result is plugged into the algorithm of Belazzougui~\cite{Belaz14} to obtain  the suffix tree topology in $O(n)$ deterministic time. 
Finally we show how the permuted LCP array (PLCP) can be constructed in $O(n)$ time, provided we already built the suffix array and  the suffix tree topology; the algorithm  is described in Section~\ref{sec:permlcp}. Our algorithm for constructing PLCP is also based on batch processing of rank queries. To make this paper self-contained we provide some background on compressed data structures and indexes in Section~\ref{sec:prelim}.

We denote by $T[i..]$ the suffix of $T$ starting at position $i$ and we denote by $T[i..j]$ the substring of $T$ that begins with $T[i]$ and ends with $T[j]$, 
$T[i..]=T[i]T[i+1]\ldots T[n-1]$ and $T[i..j]=T[i]T[i+1]\ldots T[j-1]T[j]$. We assume that the text $T$ ends with a special symbol \$ and \$ lexicographically precedes all other symbols in $T$.  The alphabet size is $\sigma$ and symbols are integers in $[0..\sigma-1]$ (so \$ corresponds to $0$).  In this paper, as in the previous papers on this topic, we use the word RAM model of computation. A machine word consists of $\log n$ bits and we can execute standard bit operations, addition and subtraction in constant time. We will assume for simplicity that the alphabet size $\sigma\le n^{1/4}$. This assumption is not restrictive because for $\sigma> n^{1/4}$ linear-time algorithms that use $O(n\log\sigma)=O(n\log n)$ bits are already known.

\section{Linear Time Construction of the Burrows-Wheeler Transform}
\label{sec:lintimebwt}
In this section we show how the Burrows-Wheeler transform (BWT) of a text $T$ can be constructed in $O(n)$ time using $O(n\log \sigma)$ bits of space.
Let $\Delta=\log_{\sigma}n$. We can assume w.l.o.g. that the text length is divisible by $\Delta$ (if this is not the case we can pad the text $T$ with $\ceil{n/\Delta}\Delta - n$ \$-symbols). The BWT of $T$ is a sequence $B$ defined as follows:
 if $T[k..]$ is the $(i+1)$-th lexicographically smallest suffix, then $B[i]=T[k-1]$\footnote{So $B[0]$ has the lexicographically smallest suffix ($i+1=1$) and
so on. The exact formula is $B[i]=T[(k-1) \mathrm{mod}\, n]$. We will write $B[i]=T[k-1]$ to avoid tedious details.}.  Thus the symbols of $B$ are the symbols that precede the  suffixes of $T$, sorted in lexicographic order. We will say that $T[k-1]$ \emph{represents} the suffix $T[k..]$ in $B$. 
Our algorithm divides the suffixes of $T$ into $\Delta$ classes and constructs   $B$ in $\Delta$ steps. We say that a suffix $S$ is a $j$-suffix for  $0\le j< \Delta$ if $S=T[i\Delta-j-1..]$ for some $i$, and denote by $\cS_j$ the set of all $j$-suffixes, $\cS_j=\{\,T[i\Delta-j-1..]\,|\, 1\le i\le n/\Delta\,\}$.  During the $j$-th step we process all $j$-suffixes and insert symbols representing  $j$-suffixes at appropriate positions of the sequence $B$. 

\paragraph{Steps $0-1$.} 
We sort suffixes in $\cS_0$ and $\cS_1$ by constructing a new text and representing it as a sequence of $n/\Delta$ meta-symbols. Let 
$T_1=T[n-1]T[0]T[1]\ldots T[n-2]$ be the text $T$ rotated by one symbol to the right  and let $T_2=T[n-2]T[n-1]T[0]\ldots T[n-3]$ be the text obtained by rotating $T_1$ one symbol to the right. We represent $T_1$ and $T_2$ as  sequences of length $n/\Delta$ over meta-alphabet $\sigma^{\Delta}$ (each meta-symbol corresponds to a string of length $\Delta$). Thus we view $T_1$ and $T_2$ as $$T_1=\boxed{T[n-1]\ldots T[\Delta-2]}\boxed{T[\Delta-1]\ldots T[2\Delta-2]}\boxed{T[2\Delta-1]\ldots T[3\Delta-2]}\boxed{T[3\Delta-1]\ldots }\ldots $$
$$T_2=\boxed{T[n-2]\ldots T[\Delta-3]}\boxed{T[\Delta-2]\ldots T[2\Delta-3]}\boxed{T[2\Delta-2]\ldots T[3\Delta-3]}\boxed{T[3\Delta-2]\ldots}\ldots $$

 Let $T_3=T_1\circ T_2$ denote the concatenation of $T_1$ and $T_2$.  To sort the suffixes of $T_3$, we sort the meta-symbols of $T_3$ and rename them with their ranks. Since meta-symbols correspond to $(\log n)$-bit integers, we can sort them in time $O(n)$ using radix sort.  Then we apply  a linear-time and linear-space suffix array construction algorithm~\cite{KarkkainenSB06} to  $T_3$.   We thus obtain a sorted list of suffixes $L$ for the meta-symbol sequence $T_3$. 
Suffixes of $T_3$ correspond to the suffixes from $\cS_0\cup \cS_1$ in the original text $T$: the suffix $T[i\Delta-1..]$ corresponds to the suffix of $\cS_0$ starting with meta-symbol $\boxed{T[i\Delta-1]T[i\Delta]\ldots}$ in $T_3$ and the suffix $T[i\Delta-2\ldots]$ corresponds to the suffix of $\cS_1$ starting with $\boxed{T[i\Delta-2]T[i\Delta-1]\ldots}$.  Since we assume that the special symbol \$ is smaller than all other symbols, this correspondence is order-preserving.  Hence by sorting the suffixes of $T_3$  we obtain the sorted list  $L'$ of suffixes in $\cS_0\cup \cS_1$.
\no{ $L$ contains the same suffixes as  $\cS_1\cup \cS_0$ with only four exceptions: Two suffixes of $T_3$, ones that start with $T[n-1]T[0]\ldots$ and $T[n-2]T[n-1]T[0]\ldots$, are not in $\cS_1\cup \cS_0$.  We can remove  two unneeded suffixes from $L$ in $O(1)$ time.  Besides two rightmost suffixes of $T$, $S_1=T[n-1]$ and $S_2=T[n-2]T[n-1]$, are not in $L$.    Since $S_1$ and $S_2$ consist of $O(1)$ symbols, we can find their positions and insert them into $L$ in $O(\log n)$ time by binary search in $L$.}  Now we are ready to insert symbols representing $j$-suffixes into $B$:  Initially $B$ is empty. Then  the list  $L'$ is traversed and  for every suffix $T[k..]$ that appears in $L'$ we add the symbol $T[k-1]$ at the end of $B$.

When suffixes in $\cS_0$ and $\cS_1$ are processed, we need to record some information for the next step of our algorithm.  
For every suffix  $S\in \cS_1$ we keep its position in the sorted list of suffixes. The position of suffix  $T[i\Delta-2..]$ is stored in the entry $W[i]$ of an auxiliary array $W$, which at the end of the $j$-th step will contain the
positions of the suffixes $T[i\Delta-j-1..]$. 
We also keep an auxiliary array $\Acc$ of size $\sigma$: $\Acc[a]$ is equal to the number of occurrences of symbols $i\le a-1$ in the current sequence $B$.

\paragraph{Step $j$ for $j\ge 2$.} 
Suppose that suffixes from $\cS_0$, $\ldots$, $\cS_{j-1}$ are already processed.
The symbols that precede suffixes from these sets are stored in the sequence $B$; the $k$-th symbol $B[k]$  in $B$ is the symbol that precedes the $k$-th lexicographically smallest  suffix from $\cup_{t=0}^{j-1}\cS_t$.  For every suffix $T[i\Delta-j..]$, we know its position $W[i]$ in $B$. Every suffix  
$S_i=T[i\Delta-j-1..]\in \cS_j$ can be represented as $S_i=aS'_i$ for some symbol $a$ and the suffix $S'_i=T[i\Delta-j..]\in \cS_{j-1}$. We look up the position $t_i=W[i]$ of $S'_i$ and answer rank query $r_i=\ra_a(t_i,B)$. We need $\Omega(\log\frac{\log \sigma}{\log\log n})$ time to answer a single rank query on a static sequence~\cite{BelazzouguiN15}.  If updates are to be supported, then we need $\Omega(\log n/\log \log n)$ time to answer such a query~\cite{FS89}.  However in our case the scenario  is different: we perform a \emph{batch} of $n/\Delta$ queries  to sequence $B$, i.e., we have to find $r_i$ for \emph{all} $t_i$. During Step $2$ the number of queries is equal to $|B|/2$ where $|B|$ denotes the number of symbols in $B$. During step $j$ the number of queries is $|B|/j\ge |B|/\Delta$.  We will show in  Section~\ref{sec:batchrank} that such a large  batch of rank queries  can be answered in $O(1)$ time per query. Now we can find  the rank  $p_i$ of $S_i$ among $\cup_{t=1}^{j}\cS_t$: there are exactly $p_i$ suffixes in  $\cup_{t=1}^{j}\cS_t$ that are smaller than $S_i$, where $p_i=\Acc[a]+r_i$. Correctness of this computation can be proved as follows.
\begin{proposition}
  \label{prop:sufrank}
Let $S_i=aS'_i$ be an arbitrary suffix from the set $\cS_j$.
For every occurrence of a symbol $a'<a$ in the sequence $B$, there is exactly one suffix $S_p< S_i$ in $\cup_{t=1}^j \cS_t$, such that $S_p$ starts with $a'$. Further, there are exactly $r_i$ suffixes $S_v$ in $\cup_{t=1}^j \cS_t$ such that $S_v\le S_i$ and $S_v$ starts with $a$.
\end{proposition}
\begin{proof}
  Suppose that  a suffix $S_p$ from $\cS_t$, such that $j\ge t\ge 1$, starts with $a'<a$. Then $S_p=a'S'_p$ for some $S'_p\in \cS_{t-1}$. By definition of the sequence $B$, 
there is exactly one occurrence of $a'$ in $B$ for every such $S'_p$. Now suppose that a suffix $S_v\in \cS_t$, such that $j\ge t\ge 1$, starts with $a$ and $S_v\le  S_i$. 
Then $S_v=aS'_v$ for $S'_v\in \cS_{t-1}$ and $S'_v\le  S'_i$. For every such $S'_v$ there is exactly one occurrence of the symbol $a$ in $B[1..t_i]$, where $t_i$ is the 
position of $S'_i$ in $B$. 
\end{proof}
The above calculation did not take into account the suffixes from $\cS_0$. We compute the number of suffixes $S_k\in \cS_0$ such that $S_k<S_i$ using the approach of Step $0-1$. Let $T_1$ be the text obtained by rotating $T$ one symbol to the right. Let $T'$ be the text obtained by rotating $T$ $j+1$ symbols to the right. We can sort suffixes of $\cS_0$ and $\cS_j$ by concatenating $T_1$ and $T'$, viewing the resulting text $T''$ as a sequence of $2n/\Delta$ meta-symbols and constructing the suffix array for $T''$. When suffixes in $\cS_0\cup \cS_j$ are sorted, we traverse the sorted list of suffixes;  for every suffix $S_i\in \cS_j$ we know the number $q_i$ of lexicographically smaller suffixes from $\cS_0$.

We then modify the sequence $B$: 
We sort new suffixes $S_i$ by $o_i=p_i+q_i$. Next we insert the symbol $T[i\Delta -j-1]$ at position $o_i-1$ in $B$ (assuming the first index of $B$ is $B[0]$); insertions are performed in increasing order of $o_i$. We will show that this procedure also takes $O(1)$ time per update for a large batch of insertions. Finally we record the position of every new suffix from $\cS_j$ in the sequence $B$. Since the positions of suffixes from $\cS_{j-1}$ are not needed any more, we use the entry $W[i]$ of $W$ to store the position of $T[i\Delta-j-1..]$. The array $\Acc$ is also updated.

When Step $\Delta-1$ is completed, the sequence $B$ contains $n$ symbols and $B[i]$ is the symbol that precedes the $(i+1)$-th smallest suffix of $T$. Thus we obtained the BWT of $T$.  Step $0$ of our algorithm uses $O((n/\Delta)\log n)=O(n\log \sigma)$ bits. For all the following steps we need to maintain the sequence $B$ and the array $W$. $B$ uses $O(\log\sigma)$ bits per symbol and $W$ needs $O((n/\Delta)\log n)=O(n\log\sigma)$ bits. Hence our algorithm uses $O(n\log \sigma)$ bits of workspace. Procedures for querying and updating $B$ are described in the following section. Our result can be summed up as follows. 
\begin{theorem}
  \label{theor:bwt}
Given a string $T[0..n-1]$ over an alphabet of size $\sigma$, we can construct 
the BWT of $T$ in $O(n)$ deterministic time using $O(n\log\sigma)$ bits. 
\end{theorem}

\section{Batched Rank Queries on a Sequence}
\label{sec:batchrank}
In this section we show how a batch  of $m$ rank queries for $\frac{n}{\log^2 n} \le m \le n$ can be answered in $O(m)$ time on a sequence $B$ of length $n$. 
We start by describing a static data structure. A data structure that supports  batches of queries and batches of insertions will be described later. 
We will assume $\sigma\ge \log^4n$; 
if this is not the case, the data structure from~\cite{FMMN07} can be used to answer rank queries in time $O(1)$.   

Following previous work \cite{GolynskiMR06}, we divide $B$ into chunks of size $\sigma$ (except for the last chunk that contains at most $\sigma$ symbols). For every symbol $a$ we keep a binary sequence  $M_a=1^{d_1}01^{d_2}0\ldots 1^{d_f}$ where $f$ is the total number of chunks and $d_i$ is the number of occurrences of $a$ in the  chunk.  We keep the following information for every chunk $C$.  Symbols in a chunk $C$ are represented as pairs $(a,i)$: we store a pair $(a,i)$ if and only if $C[i]=a$. These pairs are sorted by symbols and pairs representing the same symbol $a$ are sorted by their positions in $C$; all sorted  pairs from a chunk are kept in a sequence $R$. The array $F$ consists of $\sigma$ entries; $F[a]$ contains a pointer to the first occurrence of a symbol $a$ in $R$ (or {\em null} if $a$ does not occur in $C$). 
Let $R_a$ denote the subsequence of $R$ that contains all pairs $(a,\cdot)$ for some symbol $a$. If $R_a$ contains at least $\log^2 n$ pairs, we split $R_a$ into groups $H_{a,r}$ of size $\Theta(\log^2 n)$. For every group, we keep its first pair in the sequence $R'$. Thus $R'$ is also a subsequence of $R$.  For each pair $(a',i')$ in $R'$ we also store the partial rank of $C[i']$ in $C$, $\ra_{C[i']}(i',C)$.

All pairs in $H_{a,r}$ are kept in a data structure $D_{a,r}$ that contains  the second components of pairs $(a,i)\in H_{a,r}$. Thus $D_{a,r}$ contains positions of $\Theta(\log^2 n)$ consecutive symbols $a$. If $R_a$ contains less than $\log^2 n$ pairs, then we keep all pairs starting with symbol $a$ in one group $H_{a,0}$. Every $D_{a,r}$ contains $O(\log^2 n)$ elements. Hence we can implement $D_{a,r}$ so that predecessor queries are answered in constant time: for any integer $q$, we can find the largest $x\in H_{a,r}$ satisfying $x\le q$ in $O(1)$ time~\cite{FW94}.  We can also find the number of elements $x\in H_{a,r}$ satisfying $x\le q$ in $O(1)$ time. This operation on $H_{a,r}$ can be implemented using bit techniques  similar to those suggested in~\cite{NavarroN13}; details are to be given in the full version of this paper.

\paragraph{Queries on a Chunk.}
Now we are ready to answer a batch of queries in $O(1)$ time per query.  First we describe how queries on a chunk can be answered. Answering a query $\ra_a(i,C)$ on a chunk $C$ is equivalent to counting the number of pairs $(a,j)$ in $R$ such that $j\le i$. 
Our method works in three steps. We start by sorting the sequence of all queries on $C$.  Then we ``merge'' the sorted query sequence with $R'$. That is, we find for every $\ra_a(i,C)$ the rightmost pair $(a,j')$ in $R'$, such that $j'\le i$. Pair $(a,j')$ provides us with an approximate answer to $\ra_a(i,C)$ (up to an additive $O(\log^2 n)$ term). Then we obtain the exact answer to each query by searching in some data structure $D_{a,j}$. Since $D_{a,j}$ contains only $O(\log^2 n)$ elements, the search can be completed in $O(1)$ time.  A more detailed description follows.

 Suppose that we must answer $v$ queries $\ra_{a_1}(i_1,C)$, $\ra_{a_2}(i_2,C)$, $\ldots$, $\ra_{a_v}(i_v,C)$ on a chunk $C$. We sort the sequence of queries by pairs $(a_j,i_j)$ in increasing order. This sorting step takes $O(\sigma/\log^2 n + v)$ time, where $v$ is the number of queries: if $v<\sigma/\log^3n$, we sort in $O(v\log n)=O(\sigma/\log^2 n)$ time; if $v\ge \sigma/\log^3 n$, we sort in $O(v)$ time using radix sort (e.g., with radix $\sqrt{\sigma}$). Then we simultaneously traverse  the sorted sequence of queries and $R'$; for each query pair $(a_j,i_j)$ we identify the pair 
$(a_t,p_t)$ in $R'$ such that either (i) $p_t\le i_j\le p_{t+1}$ and $a_j=a_t=a_{t+1}$ or (ii) $p_t\le i_j$, $a_j=a_t$, and $a_t\not=a_{t+1}$.   That is, we find the largest 
$p_t\le i_j$  such that $(a_j,p_t)\in R'$ for every query pair $(a_j,i_j)$. 
If  $(a_t,p_t)$ is found, we search in the group $H_{a_t,p_t}$ that starts with the pair $(a_t,p_t)$. 
If the symbol $a_j$ does not occur in  $R'$, then we search in the leftmost group $H_{a_j,0}$. 
Using $D_{a_t,p_t}$ (resp.\ $D_{a_t,0}$), we find 
the largest position $x_t\in H_{a_t,p_t}$ such that $x_t\le i_j$. Thus $x_t$ is the largest position in $C$ satisfying $x_t\le i_j$ 
and $C[x_t]=a_j$. We can then compute $\ra_{a_t}(x_t,C)$ as follows: Let  $n_1$ be the partial rank of $C[p_t]$, 
$n_1=\ra_{C[p_t]}(p_t,C)$. Recall that we explicitly store this information for every position in $R'$.  Let $n_2$ be the number of positions 
$i\in H_{a_t,p_t}$ satisfying $i\le x_t$. We can compute $n_2$ in $O(1)$ time using $D_{a_t,p_t}$. Then $\ra_{a_j}(x_t,C)=n_1+n_2$.
Since $C[x_t]$ is the rightmost occurrence of $a_j$ up to $C[i_j]$, $\ra_{a_j}(i_j,C)=\ra_{a_j}(x_t,C)$.  
The time needed to traverse the sequence $R'$ is $O(\sigma/\log^2 n)$ for all the queries. Other computations take $O(1)$ time per query. Hence the sequence of $v$ queries on a chunk is answered in $O(v+\sigma/\log^2 n)$ time. 

\paragraph{Global Sequence.}
Now we consider the global sequence of queries $\ra_{a_1}(i_1,B)$, $\ldots$, $\ra_{a_m}(i_m,B)$.  First we assign queries to chunks (e.g., by sorting all queries by $(\floor{i/\sigma}+1)$ using radix sort). We answer the batch of queries on the $j$-th chunk in $O(m_j+\sigma/\log^2 n)$ time where $m_j$ is the number of queries on the $j$-th chunk. Since $\sum m_j=m$, all $m$ queries are answered in $O(m +n/\log^2 n)=O(m)$ time. Now we know the  rank $n_{j,2}=\ra_{a_j}(i'_j,C)$, where $i'_j=i_j-\floor{i/\sigma}\sigma$ is the relative position of $B[i_j]$ in its chunk $C$. 

 The binary sequences $M_a$ allows us reduce rank queries on $B$ to rank queries on a chunk $C$.
All sequences $M_a$ contain  $n+\floor{n/\sigma}\sigma$ bits; hence they use $O(n)$ bits of space.  We can compute the number of occurrences of $a$ in the first $j$ chunks in $O(1)$ time  by answering one select query. Consider a  rank query $\ra_{a_j}(i_j,B)$ and suppose that $n_{j,2}$ is already known.  We compute  $n_{j,1}$, where $n_{j,1}=\sel_0(\floor{i_j/\sigma},M_{a_j})-\floor{i_j/\sigma}$ is the number of times $a_j$ occurs in the first $\floor{i_j/\sigma}$ chunks. Then we compute $\ra_{a_j}(i_j,B)=n_{j,1}+n_{j,2}$.
\begin{theorem}
  \label{theor:batchstat}
We can keep a sequence $B[0..n-1]$ over an alphabet of size $\sigma$ in $O(n\log\sigma)$ bits of space so that a batch of $m$ $\ra$ queries 
can be answered in $O(m)$ time, where $\frac{n}{\log^2 n} \le m\le n$.
\end{theorem}
The static data structure of Theorem~\ref{theor:batchstat} can be dynamized so that batched 
queries and batched insertions are supported. Our dynamic data structures supports a batch of $m$ queries in time $O(m)$ and a batch of $m$ insertions in amortized time $O(m)$ for any $m$ that satisfies  $\frac{n}{\log_{\sigma}n}\le m\le n$. We describe the dynamic data structure in Sections~\ref{sec:listlabel} and~\ref{sec:batchdynseq}. 
\section{Building the Suffix Tree} 

Belazzougui proved the following result \cite{Belaz14}: if we are given the BWT $B$ of a text $T$ and if we can report all the distinct symbols in a range of $B$ in optimal time, then in $O(n)$ time we can: (i) enumerate all the suffix array intervals corresponding to internal nodes of the suffix tree and (ii) for every internal node list the labels of its children and their intervals. Further he showed that, if we can enumerate all the suffix tree intervals in $O(n)$ time, then we can build the suffix tree topology~\cite{Sadakane07} in $O(n)$ time. The algorithms need only $O(n)$ additional bits of space.  We refer to Lemmas 4 and 1 and their proofs in~\cite{Belaz14} for details. 

In Section~\ref{sec:partrank} we show that a partial rank data structure can
be built in $O(n)$ deterministic time. This can be used to build the desired
structure that reports the distinct symbols in a range, in $O(n)$ time 
and using $O(n\log\log\sigma)$ bits. The details are given in Section~\ref{sec:colrep}. Therefore, we obtain the following result.

\begin{lemma}
  \label{theor:suftreetopol}
If we already constructed the BWT of a text $T$, then we can build the suffix tree topology in $O(n)$ time using $O(n\log\log \sigma)$ additional bits. 
\end{lemma}

In Section~\ref{sec:permlcp} we show that the permuted LCP array of $T$ can be constructed in $O(n)$ time using $O(n\log\sigma)$ bits of space. Thus we obtain our main result on building compressed suffix trees. 

\begin{theorem}
  \label{theor:cst}
Given a string $T[0..n-1]$ over an alphabet of size $\sigma$, we can construct the compressed suffix tree of $T$ in $O(n)$ deterministic time using $O(n\log\sigma)$ additional bits. 
\end{theorem}

\section{Constructing the Permuted LCP Array}
\label{sec:permlcp}
The permuted LCP array is defined as $PLCP[i]=j$ if and only if $SA[r]=i$ and the longest common prefix of $T[SA[r]..]$ and $T[SA[r-1]..]$ is of length $j$. In other words $PLCP[i]$ is the length of the longest common prefix of $T[i..]$ and the suffix that precedes it in the lexicographic ordering. 
In this section we show how the permuted LCP array $PLCP[0..n-1]$ can be built in linear time.

\paragraph{Preliminaries.}
For $i=0,1,\ldots,n$ let $\ell_i=PLCP[i]$. It is easy to observe that $\ell_i\le \ell_{i+1}+1$: if the longest common prefix of $T[i..]$ and $T[j..]$ is $q$,
then the longest common prefix of $T[i+1..]$ and $T[j+1..]$ is at least $q-1$. Let $\Delta'=\Delta\log\log\sigma$ for $\Delta=\log_{\sigma}n$. By the same argument $\ell_{i}\le \ell_{i+\Delta'}+\Delta'$.  
To simplify the description we will further assume that $\ell_{-1}=0$. It can also be shown that $\sum_{i=0}^{n-1}(\ell_i-\ell_{i-1})=O(n)$.

We will denote by $B$ the BWT sequence of $T$; $\oB$ denotes the BWT of the reversed text $\oT=T[n-1]T[n-2]\ldots T[1]T[0]$.  Let $p$ be a factor (substring) of $T$ and let $c$ be a character. The operation $\extendright(p,c)$ computes the suffix interval of $pc$ in $B$ and the suffix interval of $\overline{pc}$ in $\oB$ provided that the intervals of $p$ and $\op$ are known. The operation $\contractleft(cp)$ computes the suffix intervals of $p$ and $\op$ provided that the suffix intervals of factors $cp$ and $\overline{cp}$ are known\footnote{Throughout this paper reverse strings are overscored. Thus $\overline{p}$ and $\overline{pc}$ are reverse strings of $p$ and $pc$ respectively.}.
It was demonstrated~\cite{SchnattingerOG12,BelazzouguiCKM13}
that both operations can be supported by answering $O(1)$ rank queries on $B$ and $\oB$. 

Belazzougui~\cite{Belaz14} proposed the following algorithm 
for consecutive computing of  $\ell_0$, $\ell_1$, $\ldots$, $\ell_n$.  Suppose that $\ell_{i-1}$ is already known. We already know the rank $r_{i-1}$ of $T[i-1..]$, the interval 
of $T[i-1..i+\ell_{i-1}-1]$ in $B$, and the interval of $\overline{T[i-1..i+\ell_{i-1}-1]}$ in $\oB$. We compute the rank 
$r_i$ of $T[i.. ]$. If $r_{i-1}$ is known, we can compute $r_i$ in $O(1)$ time by answering one select query on $B$; see Section~\ref{sec:prelim}. Then  we find the interval $[r_s,r_e]$ of $T[i..i+\ell_{i-1}-1]$  in $B$ and the interval $[r'_s,r'_e]$ of  $\overline{T[i..i+\ell_{i-1}-1]}$ in $\oB$.  These two intervals can be computed by $\contractleft$. In the special case when $i=0$ or $\ell_{i-1}=0$, we set $[r_s,r_e]=[r'_s,r'_e]=[0,n-1]$. Then for $j=1,2,\ldots $ we find the intervals for $T[i..i+(\ell_{i-1}-1)+j]$ and $\overline{T[i..i+(\ell_{i-1}-1)+j]}$. Every following pair of intervals is found by operation $\extendright$. We stop when 
the interval  of $T[i..i+\ell_{i-1}-1+j]$ is $[r_{s,j},r_{e,j}]$ such that $r_{s,j}=r_i$. For all $j'$,  such that $0\le j'<j$, we have   $r_{s,j'}<r_i$. It can be shown that $\ell_i=\ell_{i-1}+j-1$; see the proof of \cite[Lemma 2]{Belaz14}.  Once $\ell_i$ is computed, we increment $i$ and find the next $\ell_i$ in the same way. All $\ell_i$ are computed by $O(n)$ $\contractleft$ and $\extendright$ operations.

\paragraph{Implementing $\contractleft$ and $\extendright$.}
We create the succinct representation of the suffix tree topology both for $T$ and $\oT$; they  will be denoted by $\cT$ and $\ocT$ respectively. We keep both $B$ and $\oB$ in the data structure that supports access in $O(1)$ time. We also store $B$ in the data structure that answers $\sel$ queries in $O(1)$ time. The array $\Acc$ keeps information about accumulated frequencies of symbols: $\Acc[i]$ is the number of occurrences  of all symbols $a\leq i-1$ in $B$. Operation $\contractleft$ is implemented as follows.   Suppose that we know the interval $[i,j]$ for a factor $cp$ and the interval $[i',j']$ for the factor $\overline{cp}$.  We can compute the interval $[i_1,j_1]$ of $p$ by finding $l=\sel_c(i-\Acc[c],B)$ and $r=\sel_c(j-\Acc[c],B)$. Then we find the lowest common ancestor $x$ of leaves $l$ and $r$ in the suffix tree $\cT$. We set $i_1=\mathtt{leftmost\_leaf}(x)$ and $j_1=\mathtt{rightmost\_leaf}(x)$.  Then we consider the number of distinct symbols in $B[i_1..j_1]$. If $c$ is the only  symbol that occurs in $B[i_1..j_1]$, then all factors $p$ in $T$ are preceded by $c$. Hence all factors $\overline{p}$ in $\oT$ are followed by $c$ and $[i'_1,j'_1]=[i',j']$.  Otherwise we find the lowest common ancestor $y$ of leaves $i'$ and $j'$ in $\ocT$. Then we identify $y'=\mathtt{parent}(y)$ in $\ocT$ and let $i'_1=\mathtt{leftmost\_leaf}(y')$ and $j'_1=\mathtt{rightmost\_leaf}(y')$.   Thus $\contractleft$ can be supported in $O(1)$ time.

Now we consider the operation $\extendright$. Suppose that $[i,j]$ and $[i',j']$ are intervals of $p$ and $\overline{p}$ in $B$ and $\oB$ respectively. We compute the interval of $\overline{pc}$ by using the standard BWT machinery.  Let $i'_1=\ra_c(i'-1,\oB)+\Acc[c]$ and $j'_1=\ra_c(j',\oB)+\Acc[c]-1$. 
We check whether $c$ is the only symbol in $\oB[i'..j']$. If this is the case, then 
all occurrences of $\op$ in $\oT$ are preceded by $c$ and all occurrences of $p$ in $T$ are followed by $c$. Hence the interval of $pc$ in $B$ is 
$[i_1,j_1]=[i,j]$. Otherwise there is at least one other symbol besides $c$ that can follow $p$.  Let $x$ denote the lowest common ancestor of leaves $i$ and $j$. 
If $y$ is the child of $x$ that is labeled with $c$, then the interval of $pc$ is $[i_1,j_1]$ where  $i_1=\mathtt{leftmost\_leaf}(y)$ and  $j_1=\mathtt{rightmost\_leaf}(y)$.

We can find the child $y$ of $x$ that is labeled with $c$ by answering rank and select queries on two additional sequences, $L$ and $D$.
The sequence $L$ contains labels of children for all nodes of $\cT$; labels are ordered by nodes and labels of the same node are ordered lexicographically. We  encode the degrees of all nodes in a sequence $D=1^{d_1}01^{d_2}0\ldots 1^{d_n}$, where $d_i$ is the degree of the $i$-th node.  
We compute $v=\sel_0(x,D)-x$, $p_1=\ra_{c}(v,L)$, $p_2=\sel_c(p_1+1,L)$, and $j=p_2-v$. Then $y$ is  the $j$-th child of $x$. 
The bottleneck of $\extendright$ are the computations of $p_1$, $i_1'$, and $j_1'$ because we need $\Omega(\log\frac{\log \sigma}{\log\log n})$ time to answer a rank query on $L$ (resp.\ on $\oB$); all other calculations can be executed in $O(1)$ time.

\paragraph{Our Approach.}
Our algorithm follows the technique of~\cite{Belaz14}  that relies on operations $\extendright$ and $\contractleft$ for building the PLCP.  We implement these two operations as described above;  hence we will have  to perform $\Theta(n)$  $\ra$ queries on sequences $L$ and $\oB$. Our method creates large batches of queries; each query in a batch is answered in $O(1)$ time using  Theorem~\ref{theor:batchstat}.   

During the pre-processing stage we create the machinery for supporting 
operations $\extendright$ and $\contractleft$. We compute the BWT  $B$ of $T$ and the BWT $\oB$ for the reverse text $\oT$. We also construct the suffix tree topologies $\cT$ and $\ocT$. 
When $B$ is constructed, we record the positions in $B$ that correspond to suffixes $T[i\cdot\Delta'..]$ for $i=0,\ldots, \floor{n/\Delta'}$. PLCP construction is divided into three stages: first we compute the values of $\ell_i$ for selected evenly spaced indices $i$,  $i=j\cdot \Delta'$ and $j=0,1$,$\ldots$,$\floor{n/\Delta'}$. We use a slow algorithm for computing lengths that takes $O(\Delta')$ extra time for every $\ell_i$.  During the second stage we compute all remaining values of $\ell_i$. We use the method from~\cite{Belaz14} during Stage 2. The key to a fast implementation  is  ``parallel''  computation.  We divide all lengths into groups and assign each group of lengths to a \emph{job}. At any time we process a list containing at least $2n/\log^2 n$ jobs.  We answer $\ra$ queries in batches: when a job $J_i$ must answer a slow $\ra$ query on $L$ or $\oB$, we pause $J_i$ and add the rank query to the corresponding pool of queries. When a pool of queries on $L$ or the pool of queries on $\oB$ contains $n/\log^2 n$ items, we answer the batch of queries in $O(n/\log^2 n)$ time. The third stage starts when the number of jobs becomes smaller than  $2n/\log^2n$.   All lengths that were not computed earlier are computed during Stage 3 using the slow algorithm.  Stage 2 can be executed in $O(n)$ time because $\ra$ queries are answered in $O(1)$ time per query. Since the number of lengths that we compute during the first and the third stages is small, Stage 1 and Stage 3 also take time $O(n)$.  A more detailed description follows. 

\paragraph{Stage 1.}
Our algorithm starts by computing $\ell_i$ for $i=j\cdot\Delta'$ and $j=0,1,\ldots, \floor{n/\Delta'}$. Let $j=0$ and $f=j\Delta'$.  We already know the rank $r_f$ of $S_f=T[j\Delta'..]$ in $B$ ($r_f$ was computed and recorded when $B$ was constructed). We can also find the starting position $f'$ of the suffix $S'$  of rank $r_f-1$, $S'=T[f'..]$. Since $f'$ can be found by  employing the function LF at most $\Delta'$ times, we can compute $f'$ in $O(\Delta')$ time; see Section~\ref{sec:prelim}\footnote{A faster computation is possible, but we do not need it here.}.  When $f$ and $f'$ are known, we scan $T[f..]$
and $T[f'..]$ until the first symbol $T[f+p_f]\not=T[f'+p_f]$ is found. By definition of $\ell_j$,  $\ell_0=p_f-1$.  Suppose that 
$\ell_{s\Delta'}$ for $s=0$, $\ldots$, $j-1$ are already computed and we have to compute $\ell_f$ for $f=j\Delta'$ and some $j\ge 1$.  We already know the rank $r_f$ of suffix $T[f..]$. We find $f'$ such that the suffix $T[f'..]$ is of rank $r_f-1$  in time $O(\Delta')$. 
We showed above  that $\ell_f\ge \ell_{(j-1)\Delta'}-\Delta'$. Hence the first $o_f$ symbols in $T[f..]$ and $T[f'..]$ are equal, where $o_f=\max(0,\ell_{(j-1)\Delta'}-\Delta')$.  We scan $T[f+o_f..]$ and $T[f'+o_f..]$ until the first symbol $T[f+o_f+p_f]\not=T[f'+o_f+p_f]$ is found. By definition, $\ell_f=o_f+p_f$. Hence we compute $\ell_f$ in $O(\Delta'+ p_f)$ time for $f=j\Delta'$ and $j=1$, $\ldots$, $\floor{n/\Delta'}$. It can be shown that $\sum_f p_f=O(n)$. Hence the total time needed to compute all selected $\ell_f$ is $O((n/\Delta')\Delta'+ \sum_f p_f)=O(n)$.    For every $f=j\Delta'$ we also compute the interval of $T[j\Delta'..j\Delta'+\ell_f]$ in $B$ and the interval of $\overline{T[j\Delta'..j\Delta'+\ell_{f}]}$ in $\oB$. We show in  Section~\ref{sec:intervals} that all needed intervals can be computed in $O(n)$ time.

\paragraph{Stage 2.}
We divide $\ell_i$ into groups of size $\Delta'-1$ and compute the values of $\ell_k$ in every group using a \emph{job}. The $i$-th group contains lengths $\ell_{k+1}$, $\ell_{k+2}$, $\ldots$, $\ell_{k+\Delta'-1}$ for $k=i\Delta'$ and $i=0,1,\ldots$. All $\ell_k$ in the $i$-th group will be computed by the $i$-th \emph{job} $J_i$. Every $J_i$ is either active or paused.  Thus originally we start with a list of $n/\Delta'$ jobs and all of them are active.  All active jobs are executed at the same time. That is, we scan the list of active jobs,  spend $O(1)$ time on every active job, and then move on to the next job. When a job must answer a rank query, we pause it and insert the query into a \emph{query list}.
There are two query lists: $Q_l$ contain rank queries on sequence $L$ and $Q_b$ contains rank queries on $\oB$.  When $Q_l$ or $Q_b$ contains $n/\log^2 n$ queries, we answer all queries in $Q_l$ (resp.\ in $Q_b$). The batch of queries is answered using Theorem~\ref{theor:batchstat}, so that every query is answered in $O(1)$ time. Answers to queries are returned to jobs, corresponding jobs are re-activated,  and we continue scanning the list of active jobs.  When all $\ell_k$ for  $i\Delta' \le k <(i+1)\Delta'$ are computed, the $i$-th job is finished; we remove this job from the pool of jobs and decrement by $1$ the number of jobs. See Fig.~\ref{fig:jobs-ex}.
\begin{figure}[tb]
  \centering
  \includegraphics[height=.25\textheight]{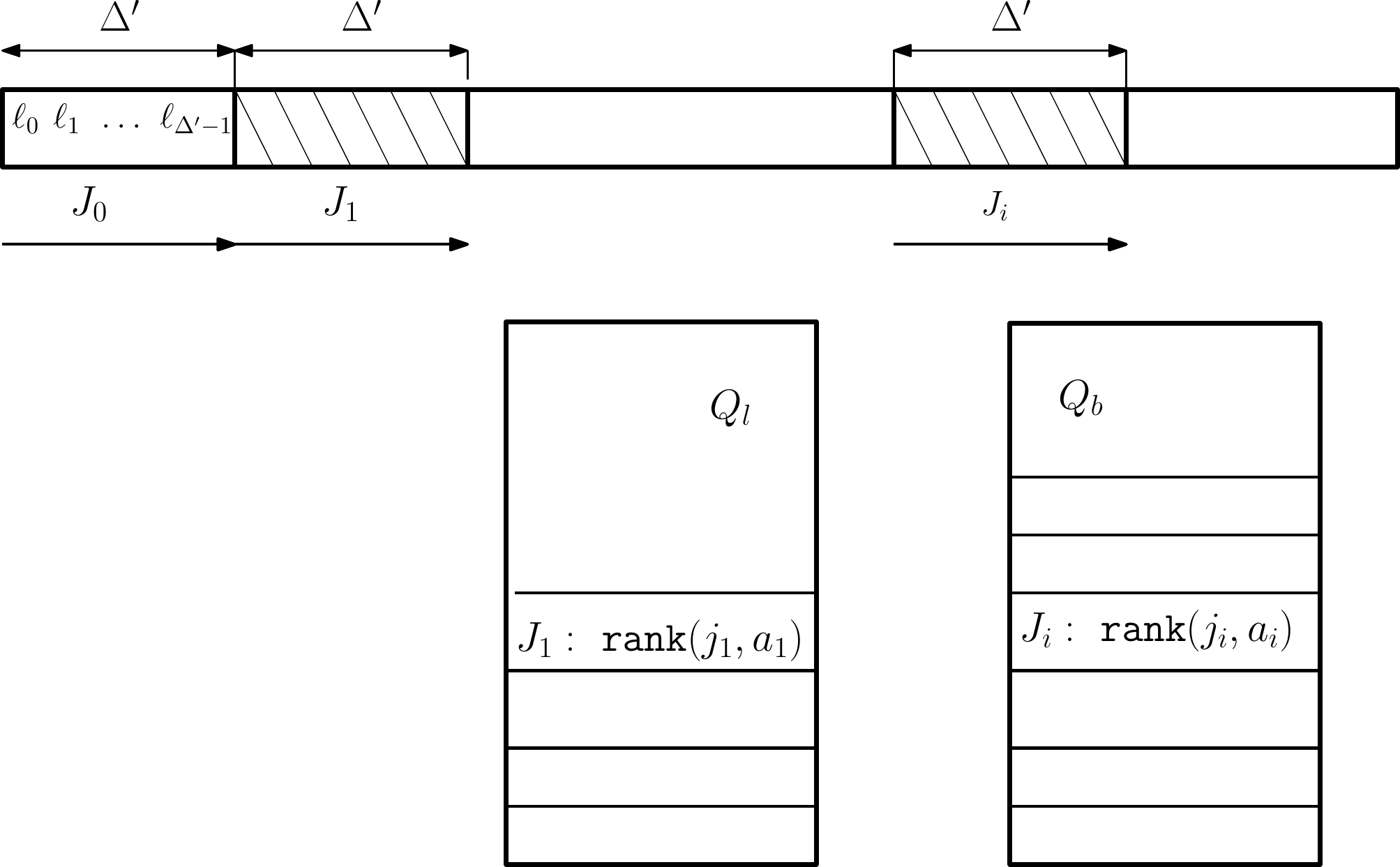}
  \caption{Computing lengths during Stage 2. Groups corresponding to paused jobs are shown shaded by slanted lines. Only selected groups are shown. The $i$-th job $J_i$ is paused because we have to answer a $\ra$ query on $\oB$; the job $J_1$ is paused because we have to answer a $\ra$ query on $L$. When $Q_l$ or $Q_b$ contains $n/\log^2 n$ queries, we answer a batch of $\ra$ queries contained in $Q_l$ or $Q_b$.}
  \label{fig:jobs-ex}
\end{figure}

Every job  $J_i$ computes $\ell_{k+1}$, $\ell_{k+2}$, $\ldots$, $\ell_{k+\Delta'-1}$ for $k=i\Delta'$ using the algorithm of Belazzougui~\cite{Belaz14}. When the interval of $T[i+\ell_k..]$ in $B$ and the interval of $\overline{T[i+\ell_k..]}$  in $\oB$ are known, we compute $\ell_{k+1}$.  The procedure for computing $\ell_{k+1}$ must execute   one operation $\contractleft$  and $\ell_{k+1}-\ell_{k}+1$ operations $\extendright$.  Operations $\contractleft$  and $\extendright$  are implemented as described above. We must answer two rank queries on $\oB$ and one rank query on $L$ for every $\extendright$.  Ignoring the time for these three rank queries, $\extendright$ takes constant time. Rank queries on $\oB$ and $L$ are answered in batches, so that each $\ra$ query takes $O(1)$ time.  Hence every operation $\extendright$ needs $O(1)$ time. The job $J_i$ needs $O(\ell_{i\Delta'+j}-\ell_{i\Delta'}+j)$ time to compute $\ell_{i\Delta'+1}$, $\ell_{i\Delta'+1}$, $\ldots$, $\ell_{i\Delta'+j}$. All $J_i$ are executed in $O(n)$ time.
   
\paragraph{Stage 3.}
``Parallel processing'' of jobs  terminates when the number of jobs in the pool becomes smaller than $2n/\log^2n$. Since every job computes $\Delta'$ values of $\ell_i$, there are at most $2n(\log\log\sigma/(\log n\log\sigma))< 2n/\log n$ unknown values of $\ell_i$ at this point. We then switch to the method of Stage 1 to compute the values of unknown $\ell_i$. All remaining $\ell_i$ are sorted by $i$ and processed in order of increasing $i$. For every unknown $\ell_i$ we compute the rank $r$ of $T[i .. ]$ in $B$. For the suffix $S'$ of rank $r-1$ we find its starting position $f'$ in $T$, $S'=T[f'..]$. Then we scan $T[f'+\ell_{i-1}-1..]$ and $T[i+\ell_{i-1}-1..]$ until the first symbol $T[f'+\ell_{i-1}+j-1]\not=T[f+\ell_{i-1}+j-1]$ is found. We set $\ell_i=\ell_{i-1}+j-2$ and continue with the next unknown $\ell_i$.  We spend $O(\Delta'+\ell_i)$ additional time for every remaining $\ell_i$; hence the total time needed to compute all $\ell_i$ is $O(n+(n/\log n)\Delta')=O(n)$.

Every job during Stage 2 uses $O(\log n)$ bits of workspace. The total number of jobs in the job list does not exceed $n/\Delta'$. The total number of queries stored at any time in lists $Q_l$ and $Q_b$ does not exceed $n/\log^2n$. Hence our algorithm uses $O(n\log \sigma)$ bits of workspace.  
\begin{lemma}
\label{lemma:permlcp}
  If the BWT of a string $T$ and the suffix tree topology for $T$ are already known, then we can compute the permuted LCP array in $O(n)$ time and $O(n\log\sigma)$ bits.  
\end{lemma}

\section{Conclusions}
\label{sec:concl}
We have shown that the Burrows-Wheeler Transform (BWT), the Compressed Suffix 
Array (CSA), and the Compressed Suffix Tree (CST) can be built in deterministic
$O(n)$ time by an algorithm that requires $O(n\log\sigma)$ bits of working 
space. Belazzougui independently developed an alternative 
solution, which also builds within the same time and space the 
simpler part of our structures, that is, the BWT and the CSA, but not the CST. His solution, that  uses different techniques, is described in the updated version of his ArXiV report~\cite{Belaz14arx} that extends his conference  paper~\cite{Belaz14}. 

Our results have  many interesting applications. For example, we can now construct an FM-index \cite{FerraginaM05,FMMN07} in $O(n)$ deterministic time using $O(n\log\sigma)$ bits. Previous results need $O(n\log\log \sigma)$ time or rely on randomization~\cite{HonSS09,Belaz14}. Furthermore Theorem~\ref{theor:partrank} enables us to support the function LF in $O(1)$ time on an FM-index. In Section~\ref{sec:index} we describe a new index based on these ideas.

Another application is that we can now compute the Lempel-Ziv 77 and 78 parsings 
\cite{LZ76,ZL77,ZL78} of a string $T[0..n-1]$ in deterministic linear time  using
$O(n\log\sigma)$ bits: K{\"o}ppl and Sadakane \cite{KS16} recently showed that,
if one has a compressed suffix tree on $T$, then they need only $O(n)$ 
additional (deterministic) time and $O(z\log n)$ bits to produce the parsing,
where $z$ is the resulting number of phrases. Since $z \le n/\log_\sigma n$,
the space is $O(n\log\sigma)$ bits. With the suffix tree, they need to compute
in constant time any $\Psi(i)$ and to move in constant time from a suffix tree
node to its $i$-th child. The former is easily supported as the inverse of the
LF function using constant-time select queries on $B$ \cite{GolynskiMR06};
the latter is also easily obtained with current topology representations
using parentheses \cite{NS14}. 

Yet another immediate application of our algorithm are index data structures for dynamic document collections. If we use our compressed index, described in Section~\ref{sec:index}, and apply Transformation 2 from~\cite{MunroNV15}, then we obtain an index data structure for a dynamic collection of documents that uses $nH_k + o(n\log\sigma)+O(n\frac{\log n}{s})$ bits where $H_k$ is the $k$-th order entropy and $s$ is a parameter. This index can count how many times a query pattern $P$ occurs in a collection in $O(|P|\log\log n + \log\log \sigma\log\log n)$ time; every occurrence can be then reported in time $O(s)$. An insertion or a deletion of some document $T_u$  is supported in $O(|T_u|\log^{\eps}n)$ and $O(|T_u|(\log^{\eps}n+s))$ deterministic time respectively. 

We believe that our technique can also  improve upon some of the recently presented results on bidirectional FM-indices~\cite{SchnattingerOG12,BelazzouguiCKM13} and other scenarios where compressed suffix trees are used~\cite{BelazzouguiCKM16}.
\paragraph{Acknowledgment.} The authors wish to thank an anonymous reviewer of this paper for careful reading and helpful comments.

\newpage

\bibliographystyle{abbrv}
\bibliography{lintime-index}

\newpage
\appendix
\renewcommand\thesection{A.\arabic{section}}

\section{Preliminaries}
\label{sec:prelim}
\paragraph{Rank and Select Queries}
The following two kinds of queries play a crucial role in compressed indexes and other succinct data structures. Consider a sequence $B[0..n-1]$ of symbols over an alphabet of size $\sigma$. The rank query $\ra_a(i,B)$ counts how many times $a$ occurs among the first $i+1$ symbols 
in $B$, $\ra_a(i,B)=|\{\,j\,|\, B[j]=a \text{ and } 0\le j< i\,\}|$.  The select query $\sel_a(i,B)$ finds the position in $B$ where $a$ occurs for the $i$-th time, $\sel_a(i,B)=j$ where $j$ is such that  $B[j]=a$ and $\ra_a(j,B)=i$. The third kind of query is the access query, $\acc(i,B)$, which returns the $(i+1)$-th symbol in $B$, $B[i]$.
If insertions and deletions of symbols in $B$ must be supported, then both kinds of queries require $\Omega(\log n/\log\log n)$ time~\cite{FS89}.  If the sequence $B$ is static, then we can answer select queries in $O(1)$ time and the cost of $\ra$ queries is reduced to $\Theta(\log\frac{\log \sigma}{\log\log n})$ \cite{BelazzouguiN15}.\footnote{If we aim to use $n\log\sigma +o(n\log\sigma)$ bits, then either $\sel$ or $\acc$ must cost $\omega(1)$. If, however, $(1+\epsilon)n\log\sigma$ bits are available, for any constant $\epsilon>0$, then we can support both queries in $O(1)$ time.} One important special case of $\ra$ queries is the partial rank query, $\ra_{B[i]}(i,B)$. Thus a partial rank query asks how many times $B[i]$ occurred in $B[0..i]$. Unlike general rank queries, partial rank queries can be answered in $O(1)$ time~\cite{BelazzouguiN15}.  In Section~\ref{sec:partrank} we describe a data structure for partial rank queries that can be constructed in $O(n)$ deterministic time. 
Better results can be achieved in the special case when the alphabet size is $\sigma=\log^{O(1)}n$; in this case  we can represent $B$ so that $\ra$, $\sel$, and $\acc$ queries are answered in $O(1)$ time~\cite{FMMN07}.

\paragraph{Suffix Tree and Suffix Array.}
A suffix tree for a string $T[0..n-1]$ is a compacted tree on the suffixes of $T$. 
The suffix array is an array $SA[0..n-1]$ such that $SA[i]=j$ if and only if $T[j..]$ is the $(i+1)$-th lexicographically smallest suffix of $T$. All occurrences of a substring $p$ in $T$ correspond to suffixes of $T$ that start with $p$; these suffixes occupy a contiguous interval in the suffix array $SA$.

\paragraph{Compressed Suffix Array.}
A compressed suffix array (CSA) is a compact data structure that provides the same functionality as the suffix array. The main component of CSA is the function $\Psi$, defined by the equality $SA[\Psi(i+1)]=(SA[i]+1)\!\!\mod n$. It is possible to regenerate the suffix array  from $\Psi$. We refer to~\cite{NM06}  and references therein for a detailed description of CSA and for  trade-offs between space usage and access time.
\paragraph{Burrows-Wheeler Transform and FM-index.}
The Burrows-Wheeler Transform (BWT) of a string $T$ is obtained by sorting all possible rotations of $T$ and writing the last symbol of every rotation (in sorted order). The BWT is related to the suffix array as follows: $BWT[i]=T[(SA[i]-1)\!\!\mod n]$. Hence, we can build the BWT by sorting the suffixes and writing the symbols that precede the suffixes in lexicographical order. This method is used in Section~\ref{sec:lintimebwt}.

The FM-index uses the BWT for efficient searching in $T$. It consists of the following three main components:
\begin{itemize}
\item 
The BWT of $T$.
\item
The array $\Acc[0..\sigma-1]$ where $\Acc[i]$ holds the total number of symbols $a\leq i-1$ in $T$ (or equivalently, the total number of symbols $a\leq i-1$ in $B$).
\item
A sampled array $SAM_b$ for a sampling factor $b$: $SAM_b$ contains values of $SA[i]$ if and only if $SA[i]\!\!\mod b =0$ or $SA[i]=n-1$.  
\end{itemize}

The search for a substring $P$ of length $m$ is performed backwards: for $i=m-1,m-2,\ldots$, we identify the interval of $p[i..m]$ in the BWT. Let $B$ denote the BWT of $T$.  
Suppose that we know the interval $B[i_1..j_1]$ that corresponds to $p[i+1..m-1]$. Then the interval $B[i_2..j_2]$ that corresponds to $p[i..m-1]$ is computed as $i_2=\ra_c(i_1-1,B)+\Acc[c]$ and $j_2=\ra_c(i_2,B)+\Acc[c]-1$, where $c=P[i]$.  Thus the interval of $p$ is found by answering $2m$ $\ra$ queries. We observe that the interval of $p$ in $B$ is exactly the same as the interval of $p$ in the suffix array $SA$.

Another important component of an FM-index is the function $LF$, defined as follows: if $SA[j]=i+1$, then $SA[LF(j)]=i$. $LF$ can be computed by answering $\ra$ queries on $B$. Using $LF$ we can find the starting position of the $r$-th smallest suffix, $SA[r]$, in $O(b)$ applications of $LF$, where $b$ is the sampling factor; we refer to~\cite{NM06} for details. It is also possible to compute the function $\Psi$ by using $\sel$ queries the BWT~\cite{LeeP07}. Therefore the BWT can be viewed as a variant of the CSA.  Using $\Psi$ we can consecutively obtain positions of suffixes $T[i..]$ in the suffix array: Let $r_i$ denote the position of $T[i..]$ in $SA$. Since $T[n-1..]=\$$ is the smallest suffix, $r_0=\Psi(0)$. For $i\ge 1$, $r_i=\Psi(r_{i-1})$ by definition of $\Psi$. Hence we can consecutively compute each $r_i$ in $O(1)$ time if we have constant-time $\sel$ queries on the BWT.

\paragraph{Compressed Suffix Tree.}
A compressed suffix tree consists of the following components:
\begin{itemize}
\item
The compressed suffix array of $T$. We can use the FM-index as an implementation.   
\item
The suffix tree topology. This component can be stored in $4n+o(n)$ bits~\cite{Sadakane07}.
\item
The permuted LCP array, or PLCP. The longest common prefix array
$LCP$ is defined as follows: $LCP[r]=j$ if and only if the longest common prefix between the suffixes of rank $r$ and $r-1$ is of length $j$.  The permuted LCP array is defined as follows: $PLCP[i]=j$ if and only if the rank of $T[i..]$ is $r$ and $LCP[r]=j$. A careful implementation of $PLCP$ occupies $2n+o(n)$ bits~\cite{Sadakane07}.
\end{itemize}

\section{Monotone List Labelling with Batched Updates}
\label{sec:listlabel}
A direct attempt to dynamize the data structure of Section~\ref{sec:batchrank} encounters one significant difficulty.  An insertion of a new symbol $a$ into a chunk $C$ changes the positions of all the symbols that follow it. Since symbols are stored in pairs $(a_j,i)$ grouped by symbol, even a single insertion into $C$ can lead to a linear number of updates. Thus it appears that we cannot support the batch of  updates on $C$ in less than $\Theta(|C|)$ time.  In order to overcome this difficulty we employ a monotone labeling method and  assign labels to positions of symbols. Every position $i$ in the chunk is assigned an integer label $\lab(i)$ satisfying $0\le\lab(i)\le \sigma\cdot n^{O(1)}$ and $\lab(i_1)< \lab(i_2)$ if and only if $i_1< i_2$.  Instead of  pairs $(a,i)$ the sequence $R$ will contain pairs $(a,\lab(i))$. 

When a new element is inserted, we have to change the labels of some other elements in order to maintain the monotonicity of the labeling. Existing labeling schemes~\cite{Willard92,BCDFCZ02,DS87} require $O(\log^2 n)$ or $O(\log n)$ changes of labels after every insertion. In our case, however, we have to process large batches of insertions. We can also assume that at most $\log n$ batches need to be processed. In our scenario  $O(1)$ amortized  modifications per insertion can be achieved, as shown below. 

In this section we denote by $C$ an ordered  set that contains between $\sigma$ and $2\sigma$ elements. Let $x_1\le x_2\le\ldots\le x_t$ denote the elements of $C$. Initially we assign the label $\lab(x_i)=i\cdot d$ to the $i$-th smallest element $x_i$, where $d=4n$.  We associate an interval $[\lab(x_i),\lab(x_{i+1})-1]$ with $x_i$. Thus initially the interval of $x_i$ is   $[id,(i+1)d-1]$. We assume that $C$ also contains a dummy element $x_0=-\infty$ and $\lab(-\infty)=0$. Thus all labels are non-negative integers bounded by $O(\sigma\cdot n)$. 

Suppose that the $k$-th batch of insertions consists of  $m$ new elements $y_1\le y_2\le\ldots\le y_m$.  Since at most  $\log n$ batches of insertions must be supported, $1\le k\le\log n$. We say that an element $y_j$ is in an interval $I=[\lab(x_s),\lab(x_e)]$ if $x_s< y_j < x_e$. We denote by $new(I)$ the number of inserted elements in $I$. The parameter  $\rho(I)$ for an interval $I$ is defined as the ratio of old and new  elements in $I=[\lab(x_s),\lab(x_e)]$, $\rho(I)=\frac{e-s+1}{new(I)}$.  We identify the set of non-overlapping intervals $I_1$, $\ldots$, $I_r$ such that  every new element $y_t$ is in some interval $I_j$,  and  $1 \le \rho(I_j) \le 2$ for all $j$,  $1\le j\le r$. (This is always possible if $m \le |C|$; otherwise we simply merge the insertions with $C$ in $O(|C|+m)=O(m)$ time and restart all the labels.) We can find $I_1$, $\ldots$, $I_r$ in $O(m)$ time. For every $I_j$, $1\le j\le r$, we evenly distribute the labels of old and  new elements in the interval $I'_j\subseteq I_j$. Suppose that $f$ new elements $y_p$, $\ldots$, $y_{p+f-1}$ are inserted into interval $I_j=[\lab(x_s),\lab(x_e)]$ so that now  there are $v=f+(e-s)+1$ elements in this interval.  We assign the label $\lab(x_s)+ d_j\cdot(i-1)$ to the $i$-th smallest element in $I_j$ where
$d_j=\frac{\lab(x_e)-\lab(x_s)}{v-1}$.   By our choice of $I_j$,  $f\le e-s+1$ and the number of elements in $I_j$ increased at most by twofold.  Hence the minimal distance between two consecutive labels does not decrease by more than a factor of $2$ after insertion of new elements into $I_j$. We inserted $f$ new elements into $I_j$ and changed the labels of at most $2f$ old elements. Hence the amortized number of labels that we must change after every insertion is $O(1)$. The initial  distance between labels is $d=4n$ and this distance becomes at most two times smaller after every batch of insertions. Hence the distance between consecutive labels is an integer larger than 2 during the  first $\log n$ batches.

One remaining problem with our scheme is the large range of the labels. Since labels are integers bounded by $4|C|n$, we need $\Theta(\log\sigma+\log n)$ bits per label.  To solve this problem, we will split the  chunk $C$ into blocks and assign the same label to all the symbols in a block. A label assigned to the symbols in a block will be stored only once.  Details are provided in  Section~\ref{sec:batchdynseq}.

\section{Batched Rank Queries and Insertions on a Sequence}
\label{sec:batchdynseq}

In this section we describe a dynamic data structure that supports both batches of rank queries and batches of insertions.  First we describe how queries and updates on a chunk $C$ are supported.  

The linked list $L$ contains all the symbols of $C$ in the same order as they appear in $C$. Each node of $L$ stores a block of $\Theta(\log_{\sigma}n)$ symbols, containing at most $(1/4)\log_{\sigma}n$ of them. We will identify list nodes with the blocks they contain; however, the node storing block $b$ also stores the total number of symbols in all preceding blocks and a label $\lab(b)$ for the block. Labels are assigned to blocks with the method described in Section~\ref{sec:listlabel}. The pointer to (the list node containing) block $b$ will be called $p_b$; these pointers use $O(\log\sigma)$ bits. 

We also maintain a data structure that can answer rank queries on any block. The data structure for a block supports queries and insertions in $O(1)$ time using a look-up table: Since  $\sigma\le n^{1/4}$ and the block size is $(1/4)\log_{\sigma}n$, we can keep pre-computed answers to all rank queries for all possible blocks in a table
$Tbl[0..n^{1/4}-1][0..n^{1/4}-1][0..\log_\sigma n-1]$.  The entry $Tbl[b][a][i]$ contains the answer to the query $\ra_a(i,b)$ on a block $b$. $Tbl$ contains $O(n^{1/2}\log_{\sigma}n)=o(n)$ entries and can be constructed in $o(n)$ time. Updates can be supported by a similar look-up table or by bit operations on the block $b$.

We also use sequences $R$ and $R'$, defined in Section~\ref{sec:batchrank}, but we make the following modifications. 
For every occurrence $C[i]=a$ of a symbol $a$ in $C$, the sequence $R$ contains pair $(a,p_b)$, where $p_b$ is a pointer to the block $b$ of $L$ that contains $C[i]$. Pairs are sorted by symbol in increasing order, and pairs with the same symbol are sorted by their position in $C$.  Unlike in  Section~\ref{sec:batchrank}, the chunk $C$ can be updated and we cannot maintain the exact position $i$ of $C[i]$ for all symbols in $C$; we only maintain the pointers $p_b$ in the pairs $(a,p_b)\in R$. 

Note that we cannot use block pointers for searching in $L$ (or in $C$). Instead, block labels are monotonously increasing and $\lab(b_1)< \lab (b_2)$ if the block $b_2$ follows $b_1$ in $L$.  Hence block labels will be used for searching and answering rank queries. Block labels $\lab(b)$ use $\Theta(\log n)$ bits of space, so we store them only once with the list nodes $b$ and access them via the pointers $p_b$. 

Groups $H_{a,j}$ are defined as in Section~\ref{sec:batchrank}; each $H_{a,j}$ contains all the pairs of $R$ that are between two consecutive elements of $R'_a$ for some $a$.   The data structure $D_{a,j}$ that permits searching in $H_{a,j}$ is defined as follows. Suppose that $H_{a,j}$ contains pairs $(a,p_{b_1})$, $\ldots$, $(a,p_{b_f})$. We then keep a Succinct SB-tree data structure~\cite{GrossiORR09} on $\lab(b_1)$, $\ldots$, $\lab(b_f)$.  This data structure requires $O(\log\log n)$ additional bits per label. For any integer $q$, it can find the largest block label $\lab(b_i) < q$ in $O(1)$ time or count the number of blocks $b_i$ such that $\lab(b_i) < q$ in $O(1)$ time (because our sets $H_{a,r}$ contain a logarithmic number of elements). The search procedure needs to access one block label, which we read from the corresponding block pointer. 

\paragraph{Queries.}
Suppose that we want to answer queries $\ra_{a_1}(i_1,C)$, $\ra_{a_2}(i_2,C)$, $\ldots$, $\ra_{a_t}(i_t,C)$ on a chunk $C$.
We traverse all the blocks of $L$ and find for every $i_j$ the label $l_j$ of the block $b_j$ that contains the $i_j$-th symbol, $l_j=\lab(b_j)$. We also compute $r_{j,1}=\ra_{a_j}(i'_j,b_j)$ using $Tbl$, where $i'_j$ is the relative position of the $i_j$-th symbol in $b_j$. Since we know the total number of symbols in all the blocks that precede $b_j$, we can compute $i'_j$ in $O(1)$ time.

We then represent the queries by pairs $(a_j,l_j)$ and sort these pairs stably in increasing order of $a_j$. Then we traverse the list of query  pairs $(a_j,l_j)$ and the sequence $R'$. For every query $(a_j,l_j)$ we find the rightmost pair $(a_j,p_j)\in R'$ satisfying $\lab(p_j)\le l_j$. Let $r_{j,2}$ denote the rank of $(a_j,p_j)$ in $R_{a_j}$, i.e., the number of pairs $(a_j,i)\in R$ preceding $(a_j,p_j)$. We keep this information for every pair in $R'$ using $O(\log\sigma)$ additional bits.  Then we use the succinct SB-tree
$D_{a_j,p_j}$, which contains information about the pairs in $H_{a_j,p_j}$ (i.e., the pairs in the group starting with $(a_j,p_j)$).  
The structure finds in constant time the largest $\lab(b_g)\in D_{a_j,p_j}$ such that $\lab(b_g)< l_j$, as well as the number $r_{j,3}$ of pairs from the beginning of $H_{a_j,p_j}$ up to the pair with label $\lab(b_g)$. The answer to the $j$-th rank query is then $\ra_{a_j}(i_j,C)=r_{j,1}+r_{j,2}+r_{j,3}$.

The total query time is then $O(\sigma/\log_\sigma n + t)$.

\paragraph{Insertions.} 
Suppose that symbols $a_1$, $\ldots$, $a_t$ are to be inserted at positions $i_1$, $\ldots$, $i_t$, respectively.  We traverse the list $L$ and identify the nodes where new symbols must be inserted. We simultaneously update the information about the number of preceding elements, for all nodes. All this is done in time $O(\sigma/\log_\sigma n + t)$.
We also perform the insertions into the blocks. If, as a result, some block contains more than $(1/4)\log_{\sigma} n$ symbols, we split it into an appropriate 
number of blocks, so that each block contains $\Theta(\log_{\sigma}n)$ but at most $(1/4)\log_{\sigma}n$ symbols. 
Nodes for the new blocks are allocated\footnote{Constant-time allocation is
possible because we use fixed-size nodes, leaving the maximum possible space, 
$(1/4)\log n$ bits, for the block contents.}, linked to the list $L$, and 
assigned appropriate labels using the method described in 
Section~\ref{sec:listlabel}.  
After $t$ insertions, we create at most $O(t/\log_{\sigma}n)$ new blocks (in the
amortized sense, i.e., if we consider the insertions from the beginning). Each
such new block $b'$, coming from splitting an existing block $b$, requires that
we change all the corresponding pointers $p_b$ from the pairs $(a_z,p_b)$ in 
$R$ (and $R'$), so that they become $(a_z,p_{b'})$. To find those pairs
efficiently, the list node holding $b$ also contains the $O(\log_\sigma n)$ 
pointers to those pairs (using $O(\log\sigma)$ bits each); we can then update 
the required pointers in $O(t)$ total time.

The new blocks also require creating their labels. Those $O(t/\log_{\sigma}n)$ 
label insertions also trigger $O(t/\log_{\sigma}n)$ changes of other labels, with the technique of Section~\ref{sec:listlabel}. If the label of a block $b$ was changed, we visit all pairs $(a_z,p_{b})$ in $R$ that point to $b$. Each such $(a_z,p_b)$ is kept in some group $H_{a_z,k}$ and in some succinct SB-tree $D_{a_z,k}$. We then delete the old label of $b$ from $D_{a_z,k}$ and insert the new modified label. The total number of updates is thus bounded by $O(t)$. While not mntioned in the original paper \cite{GrossiORR09}, one can easily perform constant-time insertions and deletions of labels in a succinct SB-tree: 
The structure is a two-level B-tree of arity $\sqrt{\log n}$ holding encoded 
Patricia trees on the bits of the keys, and storing at the leaves the positions
of the keys in $H_{a,r}$ using $O(\log\log n)$ bits each. To insert or delete 
a label we follow the usual B-tree procedures. The insertion or deletion of a 
key in a B-tree node is done in constant time with a precomputed table that, 
in the same spirit of $Tbl$, yields the resulting Patricia tree if we delete 
or insert a certain node; this is possible because internal nodes store only
$O(\sqrt{\log n}\log\log n)=o(\log n)$ bits. Similarly, we can delete or insert
a key at the leaves of the tree.

Apart from handling the block overflows, we must insert in $R$ the pairs 
corresponding to the new $t$ symbols we are actually inserting. We perform
$t$ rank queries $\ra_{a_1}(i_1,C)$, $\ldots$, $\ra_{a_t}(i_t,C)$, just as 
described above, and sort the symbols to insert by those ranks using radix
sort. We then traverse $R'$ and identify the groups $H_{a_1,j_1}$, $\ldots$, 
$H_{a_t,j_t}$ where new symbols must be inserted; the counters of preceding
pairs for the pairs in $R'$ is easily updated in the way. We allocate the pairs
$(a_k,p_{b_k})$ that will belong to $H_{a_i,j_i}$ and insert the labels
$\lab(b_k)$ in the corresponding data structures $D_{a_k,j_k}$, for all
$1\le k\le t$. If some groups $H_{a_t,j_t}$ become larger than permitted, we 
split them as necessary and insert the corresponding pairs in $R'$. We can 
answer the rank queries, traverse $R$, and update the groups $H_{a_k,j_k}$ all
in $O(\sigma/\log_{\sigma}n + t)$ time. 

\paragraph{Global Sequence.}
In addition to chunk data structures, we keep a static bitvector $M_a=1^{d_1}0\ldots 1^{d_s}$ for every symbol $a$; $d_i$ denotes the number of times $a$ occurs in the $i$-th chunk.

Given a global sequence of $m\ge n/\log_{\sigma}n$  queries, $\ra_{a_1}(i_1,B)$, $\ldots$, $\ra_{a_m}(i_m,B)$ on $B$, we can assign them to chunks in $O(m)$ time. Then we answer queries on chunks as shown above.  If $m_j$ queries are asked on chunk $C_j$, then these queries are processed in $O(m_j+\sigma/\log_{\sigma}n)$ time.  Hence all queries on all chunks are answered in $O(m+n/\log_{\sigma}n)=O(m)$ time. We can answer 
a query $\ra_{a_k}(i_k,B)$ by answering a rank query on the chunk that contains $B[i_k]$ and $O(1)$ queries on the sequence $M_{a_k}$ \cite{GolynskiMR06}. Queries on $M_{a_k}$ are supported in $O(1)$ time because the bitvector is static. Hence the total time to answer $m$ queries on $B$ is $O(m)$.

When a batch of symbols is inserted, we update the corresponding chunks as described above. If some chunk contains more than $4\sigma$ symbols, we split it into several chunks of size $\Theta(\sigma)$ using standard techniques.  Finally we update the global sequences $M_a$, both because of the insertions and due to the possible chunk splits. We simply rebuild the bitvectors $M_a$ from scratch; this is easily done in $O(n_a/\log n)$ time, where $n_a$ is the number of bits in $M_a$; see e.g.~\cite{MunroNV14}. This adds up to $O(m/\log n)$ time.

Hence the total amortized cost for a batch of $m\ge n/\Delta$ insertions is $O(m)$.

\begin{theorem}
  \label{theor:batchdyn}
 We can keep a sequence $B[0..n-1]$ over an alphabet of size $\sigma$ in $O(n\log\sigma)$ bits of space so that a batch of $m$ $\ra$ queries 
can be answered in $O(m)$ time and a batch of $m$ insertions is supported in $O(m)$ amortized time, for $\frac{n}{\log_{\sigma}n}\le m\le n$.  
\end{theorem}

\section{Sequences with Partial Rank Operation}
\label{sec:partrank}
If $\sigma=\log^{O(1)}n$, then we can keep a sequence $S$ in $O(n\log\sigma)$ bits so that select and rank queries (including partial rank queries) are answered in constant time~\cite{FMMN07}.  In the remaining part of this section we will assume that 
$\sigma\ge \log^3 n$.

\begin{lemma}
  \label{lemma:partrank}
Let $\sigma\le m\le n$. We can support partial rank queries on a sequence $C[0..m-1]$ over an alphabet of size $\sigma$ in time $O(1)$. The data structure needs $O(m\log\log m)$ additional bits and can be constructed in $O(m)$ deterministic time. 
\end{lemma}
\begin{proof}
Our method employs the idea of buckets introduced in~\cite{BelazzouguiBPV09}. Our structure does not use monotone perfect hashing, however. Let $I_a$ denote the set of positions where a symbol $a$ occurs in $C$, i.e., $I_a$ contains all integers  $i$ satisfying $C[i]=a$. If $I_a$ contains more than $2\log^2m$ integers, we divide $I_a$ into buckets $B_{a,s}$  of size $\log^2 m$.  Let $p_{a,s}$ denote the longest common prefix of all integers (seen as bit strings) in the bucket $B_{a,s}$ and let $l_{a,s}$ denote the length of $p_{a,s}$. 
For every element $C[i]$ in the sequence we keep the value of $l_{C[i],t}$ where $B_{C[i],t}$ is the bucket containing $i$. If $I_{C[i]}$ was not divided into buckets, we assume $l_{C[i],t}=$ {\em null}, a dummy value.  We will show below how the index $t$ of  $B_{C[i],t}$ can be identified if $l_{C[i],t}$ is known.  For every symbol $C[i]$ we also keep the  rank $r$ of $i$ in its bucket $B_{C[i],t}$. That is, for every $C[i]$ we store the value of $r$ such that  $i$ is the $r$-th smallest element in its bucket $B_{C[i],t}$. Both $l_{C[i],t}$ and $r$ can be stored in $O(\log\log m)$ bits. The partial  rank of $C[i]$ in $C$ can be computed from $t$ and $r$,  $\ra_{C[i]}(i,C)=t\log^2m +r$.

It remains to describe how the index $t$ of the bucket containing $C[i]$ can be found. Our method uses $o(m)$ additional bits.  First we observe that $p_{a,i}\not= p_{a,j}$ for any fixed $a$ and $i\not=j$; see~\cite{BelazzouguiBPV09} for a proof. 
Let $T_w$ denote the full binary trie on the interval $[0..m-1]$.  
Nodes of $T_w$ correspond to all possible bit prefixes of integers $0,\ldots,m-1$.
We say that a bucket $B_{a,j}$ is assigned to a node $u\in T_w$ if 
$p_{a,j}$ corresponds to the node $u$. Thus  many different buckets can be assigned to the same  node $u$. But for any symbol $a$ at most one bucket $B_{a,k}$ is assigned to $u$. If a bucket is assigned to a node $u$, then there are at least $\log^2 m$ leaves below $u$.  Hence buckets can be  assigned to nodes of height at least $2\log\log m$; such nodes will be further called bucket  nodes.  We store all buckets assigned to bucket nodes of $T_w$ using the structure described below. 

We order the nodes $u$ level-by-level starting at the top of the tree. 
Let $m_j$ denote the number of buckets assigned to $u_j$. The data structure $G_j$ contains all symbols $a$ such that some bucket $B_{a,k_a}$ is assigned to $u_j$. For every symbol $a$ in $G_j$ we can find in $O(1)$ time the index $k_a$ of the bucket $B_{a,k_a}$ 
that is assigned to $u_j$. We implement $G_j$ as deterministic dictionaries of Hagerup et al.~\cite{HagerupMP01}. $G_j$ uses 
$O(m_j\log \sigma)$ bits and can be constructed in $O(m_j\log\sigma)$ time.  
We store $G_j$ only for bucket nodes $u_j$ such that $m_j>0$.  We also  keep  an  array $W[1..\frac{m}{\log^2 m}]$ whose entries correspond to bucket nodes of $T_w$: $W[j]$ contains a pointer to $G_j$ or {\em null} if $G_j$ does not exist.

Using $W$ and $G_j$ we can answer  a partial rank query $\ra_{C[i]}(i,C)$. Let $C[i]=a$.  Although the bucket $B_{a,t}$ containing $i$ is not known, we know the length $l_{a,t}$ of the prefix $p_{a,t}$. Hence $p_{a,t}$ can be computed by extracting the first $l_{a,t}$ bits of $i$. We can then find the index $j$ of the node $u_j$ that corresponds to $p_{a,t}$, $j=(2^{l_{a,t}}-1)+p_{a,t}$. We lookup the address of the data structure $G_j$ in $W[j]$. Finally the index $t$ of the bucket $B_{a,t}$ is computed as $t=G_j[a]$.    

A data structure $G_j$ consumes $O(m_j\log m)$ bits. Since $\sum_j m_j\le \frac{m}{\log^2 m}$, all $G_j$ use $O(m/\log m)$ bits of space. The array $W$ also uses $O(m/\log m)$ bits. Hence our data structure uses $O(\log\log m)$ additional bits per symbol.
\end{proof}

\begin{theorem}
  \label{theor:partrank}
We can support partial rank queries on a sequence $B$ using $O(n\log\log \sigma)$ additional bits. The underlying data structure can be constructed in $O(n)$ deterministic time. 
\end{theorem}
\begin{proof}
  We divide the sequence $B$ into chunks of size $\sigma$ (except for the last chunk that contains $n-(\floor{n/\sigma}\sigma)$ symbols).  Global sequences $M_a$ are defined in the same way as in Section~\ref{sec:batchrank}. A partial rank query on $B$ can be answered by a partial rank query on a chunk and two queries on $M_a$. 
\end{proof}

\section{Reporting All Symbols in a Range}
\label{sec:colrep}

We prove the following lemma in this section.

\begin{lemma}
\label{lemma:colrep}
Given a sequence $B[0..n-1]$ over an alphabet $\sigma$, we can build in $O(n)$ time a data structure that uses $O(n\log\log \sigma)$ additional bits and answers the following queries: for any range $[i..j]$, report $\occ$ distinct symbols that occur in $B[i..j]$ in $O(\occ)$ time, and for every reported symbol $a$, give its frequency in $B[i..j]$ and its frequency in $B[0..i-1]$.   
\end{lemma}

The proof 
is the same as that of Lemma 3 in~\cite{Belaz14}, but we use the result of Theorem~\ref{theor:partrank} to answer partial rank queries. This allows us to construct the data structure in $O(n)$ deterministic time (while the data structure in~\cite{Belaz14} achieves the same query time, but the construction algorithm requires randomization).  For completeness we sketch the proof below. 

Augmenting $B$ with $O(n)$ additional bits, we can report  all distinct symbols occurring in $B[i..j]$ in $O(\occ)$ time using the idea originally introduced by Sadakane~\cite{Sadakane07doc}.  For every reported symbol we can find in $O(1)$ time its leftmost and its rightmost occurrences  in $B[i..j]$. Suppose $i_a$ and $j_a$ are the leftmost and rightmost occurrences of $a$ in $B[i..j]$. Then the frequencies of $a$ in $B[i..j]$ and $B[0..i-1]$ can be computed as $\ra_a(j_a,B)-\ra_a(i_a,B)+1$ and $\ra_a(i_a,B)-1$ respectively.  Since $\ra_a(i_a,B)$ and $\ra_a(j_a,B)$ are partial rank queries, they are answered in $O(1)$ time. The data structure that reports the leftmost and the rightmost occurrences can be constructed in $O(n)$ time. Details and references can be found in~\cite{BelazzouguiNV13}. Partial rank queries are answered by the data structure of Theorem~\ref{theor:partrank}. Hence the data structure of Lemma~\ref{lemma:colrep} can be built in $O(n)$ deterministic time. We can also use the data structure of Lemma~\ref{theor:partrank} to determine whether the range $B[i..j]$ contains only one distinct symbol in $O(1)$ time by using the following observation. If $B[i..j]$ contains only one symbol, then $B[i]=B[j]$ and  $\ra_{B[i]}(j,B)-\ra_{B[i]}(i,B)=j-i+1$. Hence we can find out whether $B[i..j]$ contains exactly one symbol in $O(1)$ time by answering two partial rank queries. This observation will be helpful in Section~\ref{sec:permlcp}.


\section{Computing the Intervals}
\label{sec:intervals}
The algorithm for constructing PLCP, described in Section~\ref{sec:permlcp}, requires that we compute the intervals  of $T[j\Delta'..j\Delta'+\ell_i]$ and 
$\overline{T[j\Delta'..j\Delta'+\ell_{i}]}$  for $i=j\Delta'$ and $j=0,1,\ldots, n/\Delta'$.
We will show in this section how all necessary intervals can be computed in linear time when  $\ell_i$ for $i=j\Delta'$ are known. Our algorithm uses the suffix tree topology. We construct some additional data structures and pointers for selected  nodes of the suffix tree $\cT$. First, we will describe auxiliary data structures on $\cT$. Then we show how these structures can be used to find all needed intervals in linear time. 

\paragraph{Marking Nodes in a Tree.}
We use the  marking scheme described in~\cite{NavarroN12}. Let $d=\log n$.  A node $u$ of $\cT$ is \emph{heavy} if it has at least $d$ leaf descendants and \emph{light} otherwise.  We say that a heavy node $u$ is a \emph{special} or marked node if $u$ has at least two heavy children.  If a non-special heavy node $u$ has more than $d$ children and among them is one heavy child, then we keep the index of the heavy child in $u$.  

We keep all children of a node $u$ in the data structure $F_u$, so that the child of $u$ that is labeled by a symbol $a$ can be found efficiently. 
If $u$ has at most $d+1$ children, then $F_u$ is implemented as a fusion tree~\cite{FW94}; we can find the child of $u$ labeled by any symbol $a$ in $O(1)$ time. 
If $u$ has more than $d+1$ children, then  $F_u$ is implemented as the van Emde Boas data structure and we can find the child labeled by $a$ in $O(\log\log \sigma)$ time. If the node $u$ is special, we keep labels of its heavy children in the data structure $D_u$.  $D_u$ is implemented as a dictionary data structure~\cite{HagerupMP01} so that we can find any heavy child of a special node in $O(1)$ time.  We will say that a node $u$ is \emph{difficult} if $u$ is light but the parent of $u$ is heavy. We can quickly navigate from a node $u\in \cT$ to its child $u_i$  unless the node $u_i$ is difficult. 
\begin{proposition}
\label{prop:descendant}
   We can find  the child $u_i$ of $u$ that is labeled with a symbol $a$ in $O(1)$ time unless the node $u_i$ is difficult.   If $u_i$ is difficult, we can find $u_i$ in $O(\log\log \sigma)$ time. 
\end{proposition}
\begin{proof}
   Suppose that $u_i$ is heavy. If  $u$ is special, we can find $u_i$ in $O(1)$ time using $D_u$.  If $u$ is not special and it has at most $d+1$ children, then we find $u_i$ in $O(1)$ time using $F_u$. If $u$ is not special and it has more than $d+1$ children, then $u_i$ is the only heavy child  of $u$ and its index $i$ is stored with the node $u$.  Suppose that $u_i$ is light and $u$ is also light. Then $u$ has at most $d$ children and we can find $u_i$ in $O(1)$ time using $F_u$. 
If $u$ is heavy and $u_i$ is light, then $u_i$ is a difficult node. In this case we can find the index $i$ of $u_i$ in $O(\log\log \sigma)$ time using $F_u$. 
\end{proof}
\begin{proposition}
\label{prop:path}
 Any path from a node $u$ to its descendant $v$ contains at most one difficult node. 
\end{proposition}
\begin{proof}
  Suppose that a node $u$ is  a heavy node and its descendant $v$ is a light node.  Let $u'$ denote the first light node on the path from $u$ to $v$. Then all descendants of $u'$ are light nodes and $u'$ is the only difficult node on the path from $u$ to $v$. If $u$ is light or $v$ is heavy, then there are apparently no difficult nodes between $u$ and $v$. 
\end{proof}

\paragraph{Weiner Links.}
A Weiner link (or w-link) $\wlink(v,c)$ connects a node $v$ of the suffix tree $\cT$ labeled by the path $p$ to the node $u$, such that $u$ is the locus of $cp$. 
If $\wlink(v,c)=u$ we will say that $u$ is the target node and $v$ is the source of $\wlink(v,c)$ and $c$ is the label of $\wlink(v,c)$. 
If the target node $u$ is labeled by $cp$, we say that the w-link is explicit. If $u$ is labeled by some path $cp'$, such that $cp$ is a proper prefix of $cp'$, then the Weiner link is implicit. We classify Weiner links using the same technique that was applied to nodes of the suffix tree above.  Weiner links that share the same source node are called sibling links. A Weiner link from $v$ to $u$ is \emph{heavy} if the node $u$ has at least $d$ leaf descendants and \emph{light} otherwise.  A node $v$ is \emph{w-special} iff there are at least two heavy w-links connecting  $v$ and some other nodes. 
For every special node $v$ the dictionary $D'_v$ contains the labels $c$ of all heavy w-links $\wlink(v,c)$. For every $c$ such that $\wlink(v,c)$ is heavy, we also keep the target node $u=\wlink(v,c)$. $D'_v$ is implemented as in \cite{HagerupMP01} so that queries are answered in $O(1)$ time.  Suppose that $v$ is the source node of  at least $d+1$ w-links, but $u=\wlink(v,c)$ is the only heavy link that starts at $v$. In this case we say that $\wlink(v,c)$ is \emph{unique} and we store the index of $u$ and the symbol $c$ in $v$. Summing up, we store only heavy w-links that start in a w-special node or  unique w-links. All other w-links are not stored explicitly; if they are needed, we compute them using additional data structures that will be described below. 

Let $B$ denote the BWT of $T$. We split $B$ into intervals $G_j$ of size $4d^2$. 
For every $G_j$ we keep the dictionary $A_j$ of symbols that occur in $G_j$. 
For each symbol $a$ that occurs in  $G_j$, the data structure $G_{j,a}$ contains all positions of $a$ in $G_j$.  Using $A_j$, we can find out whether a symbol $a$ occurs in $G_j$. Using $G_{j,a}$, we can find for any position $i$  the smallest $i'\ge i$ such that $B[i']=a$ and $B[i']$ is in $G_j$ (or the largest $i''\le i$ such that $B[i'']=a$ and $B[i'']$ is in $G_j$).   We implement both $A_j$ and $G_{j,a}$ as fusion trees~\cite{FW94} so that queries are answered in $O(1)$ time. Data structures $A_j$ and $G_{j,a}$ for a fixed $j$ need $O(d^2\log\sigma)$ bits. We also keep (1) the data structure from~\cite{GolynskiMR06} that supports $\sel$ queries on $B$ in $O(1)$ time and rank queries on $B$ in $O(\log\log\sigma)$ time and (2) the data structure from Theorem~\ref{theor:partrank} that supports partial rank queries in $O(1)$ time.  All additional data structures on the sequence $B$ need $O(n\log\sigma)$ bits. 

\begin{proposition}
\label{prop:heavyspec}
The total number of heavy w-links that start in  w-special nodes is $O(n/d)$.    
\end{proposition}
\begin{proof}
  Suppose that $u$ is a w-special node and let $p$ be the label of $u$. Let $c_1$, $\ldots$, $c_s$ denote the labels of heavy w-links with source node $u$. This means that each $c_1p$, $c_2p$, $\ldots$, $c_sp$ occurs at least $d$ times in 
$T$.  Consider the suffix tree $\ocT$ of the reverse text $\oT$. $\ocT$ contains the node $\overline{u}$ that is labeled with $\overline{p}$.  The node $\overline{u}$ has (at least) $s$ children $\overline{u}_1$, $\ldots$, $\overline{u}_s$. The edge connecting $\overline{u}$ and $\overline{u}_i$ is a string that starts with $c_i$. In other words each $\overline{u}_i$ is the locus of $\overline{p}c_i$.  Since $c_ip$ occurs at least $d$ times in $T$, $\overline{p}c_i$ occurs at least $d$ times in $\oT$. Hence each $\overline{u}_i$ has at least $d$ descendants.  Thus every w-special node in $\cT$ correspond to a special node in $\ocT$ and every heavy w-link outgoing from a w-special node corresponds to some heavy child of a special node in $\ocT$. Since the number of heavy children of special nodes in a suffix tree is $O(n/d)$, the number of heavy w-links starting in a w-special node is also $O(n/d)$.  
\end{proof}

\begin{proposition}
\label{prop:heavy2}
The total number of unique w-links is $O(n/d)$.  
\end{proposition}
\begin{proof}
A Weiner link $\wlink(v,a)$ is unique only if $\wlink(v,a)$ is heavy, all other w-links outgoing from $v$ are light, and there are at least $d$ light  outgoing w-links from $v$. Hence there are at least $d$ w-links for every explicitly stored target node of a unique Weiner link.
\end{proof}

We say that $\wlink(v,a)$ is difficult if its target node $u=\wlink(v,a)$ is light and its source node $v$ is heavy.  
\begin{proposition}
\label{prop:wdescendant}
   We can compute $u=\wlink(v,a)$ of $u$  in $O(1)$ time unless $\wlink(v,a)$ is difficult.   If the $\wlink(v,a)$ is difficult, we can compute  $u=\wlink(v,a)$ in $O(\log\log \sigma)$ time. 
\end{proposition}
\begin{proof}
   Suppose that $u$ is heavy. If  $v$ is w-special, we can find $u$ in $O(1)$ time using $D_u$.  If $v$ is not w-special and it has at most $d+1$ w-children, then we find $u_i$ in $O(1)$ time using data structures on $B$. Let $[l_v,r_v]$ denote the suffix range of $v$. The suffix range of $u$ is $[l_u,r_u]$ where $l_u=\Acc[a]+\ra_a(l_v-1,B)+1$ and $r_u=\Acc[a]+\ra_a(r_v,B)$. We can find $\ra_a(r_v,B)$ as follows. Since $v$ has at most $d$ light w-children, the rightmost occurrence of $a$ in $B[l_v,r_v]$  is within the distance $d^2$ from $r_v$. Hence we can find the rightmost $i_a\le r_v$ such that $B[i_a]=a$ by searching in the interval $G_j$ that contains $r_v$ or the preceding interval $G_{j-1}$.  When $i_a$ is found, $\ra_a(r_v,B)=\ra_a(i_a,B)$ can be computed in $O(1)$ time because partial rank queries on $B$ are supported in time $O(1)$. We can compute $\ra_a(l_v-1,B)$ in the same way.  When $\ra$ queries are answered, we can find $l_u$ and $r_u$ in constant time.  Then we can identify the node $u$ by computing the lowest common ancestor of $l_u$-th and $r_u$-th leaves in $\cT$. 

If $v$ is not special and it has more than $d+1$ outgoing w-links, then $u$ is the only heavy target node of a w-link starting at $v$; hence, its index $i$ is stored in the node $v$.  Suppose that $u$ is light and $v$ is also light. Then the suffix range $[l_v,r_v]$ of $v$ has length at most $d$.  $B[l_v,r_v]$ intersects at most two intervals $G_j$. Hence we can find $\ra_a(l_v-1,B)$ and $\ra_a(r_v,B)$ in constant time.  Then we can find the range $[l_u,r_u]$ of the node $u$ and identify $u$ in time $O(1)$ as described above. If $v$ is heavy and $u$ is light, then $\wlink(v,a)$ is a difficult w-link. In this case we need $O(\log\log \sigma)$ time to compute $\ra_a(l_v-1,B)$ and $\ra_a(r_v,B)$. Then we find the range $[l_u,r_u]$ and the node $u$ is found as described above. 
\end{proof}
\begin{proposition}
\label{prop:wpath}
 Any sequence of nodes $u_1$, $\ldots$, $u_t$ where $u_i=\wlink(u_{i-1},a_{i-1})$ for some symbol $a_{i-1}$ contains at most one difficult w-link. 
\end{proposition}
\begin{proof}
  Let $\pi$ denote the path of w-links that  contains nodes $u_1$, $\ldots$, $u_t$.  Suppose that a node $u_1$ is  a heavy node and  $u_t$ is a light node.  Let $u_l$ denote the first light node on the path $\pi$.  Then all nodes on the path from $u_l$ to $u_t$ are light nodes and $\wlink(u_{l-1},a_{l-1})$ is the only difficult w-link  on the path from $u_1$ to $u_t$. 
If $u_1$ is light or $u_t$ is heavy, then all nodes on $\pi$ are light nodes (resp. all nodes on $\pi$ are heavy nodes). In this case  there are apparently no difficult w-links  between $u_1$ and $u_t$. 
\end{proof}

\paragraph{Pre-processing.}
Now we show how we can construct above described  auxiliary data structures in linear time. We start by generating the suffix tree topology and creating data structures $F_u$ and $D_u$ for all nodes $u$.  For every node $u$ in the suffix tree we  create the list of its children $u_i$ and their labels in $O(n)$ time. For every tree node $u$ we can  find the number of its leaf descendants using standard operations on the suffix tree topology. Hence, we can determine whether $u$ is a heavy or a light node and whether $u$ is a special node. When this information is available, we generate the data structures $F_u$ and $D_u$. 

We can create data structures necessary for navigating along w-links in a similar way.  We visit all nodes $u$ of $\cT$. Let $l_u$ and $r_u$ denote the indexes of leftmost and rightmost leaves in the subtree of $u$. Let $B$ denote the BWT of $T$. Using the method of Lemma~\ref{lemma:colrep}, we can generate the list 
of distinct symbols in $B[l_u..r_u]$ and count how many times every symbol occurred in $B[l_u..r_u]$ in $O(1)$ time per symbol. If a symbol $a$ occurred more than $d$ times, then $\wlink(u,a)$ is heavy. Using this information, we can identify w-special nodes and create data structures $D'_u$. Using the method of~\cite{Ruzic08}, we can construct $D'_u$ in 
$O(n_u\log\log n_u)$ time. By Lemma~\ref{prop:heavyspec} the total number of target nodes in all $D'_u$ is $O(n/d)$; hence we can construct all $D'_u$ in $o(n)$ time.  
 We can also find all nodes $u$ with a unique w-link. All dictionaries $D'_u$ and all unique w-links need $O((n/d)\log n)=O(n)$ bits of space.  

\paragraph{Supporting a Sequence of $\extendright$ Operations.}
\begin{lemma}
\label{lemma:extendsequence}
If we know the suffix interval of  a right-maximal factor $T[i..i+j]$  in $B$ and the suffix interval of $\overline{T[i..i+j]}$ in $\oB$, the we can find the intervals of $T[i..i+j+t]$ and $\overline{T[i..i+j+t]}$   in $O(t+\log\log\sigma)$ time.  
\end{lemma}
\begin{proof}
  Let $\cT$ and $\ocT$ denote the suffix tree for the text $T$ and let $\ocT$ denote the suffix tree of the reverse text $\oT$. 
  We keep the data structure for navigating the suffix tree $\cT$, described in Proposition~\ref{prop:descendant} and the data structure for computing Weiner links described in Proposition~\ref{prop:wdescendant}. We also keep the same data structures for $\ocT$.  
  Let $[\ell_{0,s},\ell_{0,e}]$ denote the suffix interval of $T[i..i+j]$; let $[\ell'_{0,s},\ell'_{0,e}]$ denote the suffix interval of $\overline{T[i..i+j]}$.  We navigate down the tree following the symbols $T[i+j+1]$, $\ldots$, $T[i+j+t]$.  Let $a=T[i+j+k]$
for some $k$ such that $1\le k\le t$ and suppose that the suffix interval $[\ell_{k-1,s},\ell_{k-1,e}]$ of $T[i..i+j+k-1]$ and the suffix interval $[\ell'_{k-1,s},\ell'_{k-1,e}]$ of $\overline{T[i..i+j+k-1]}$ are already known. First, we check whether our current location is a node of $\cT$. If $\oB[\ell'_{k-1,s},\ell'_{k-1,e}]$ contains only one symbol $T[i+j+k]$, then the range of $T[i..i+j+k]$ is identical with the range of $T[i..i+j+k-1]$. We can calculate the range of $\overline{T[i..i+j+k]}$ in a standard way by answering two rank queries on $\oB$ and $O(1)$ arithmetic operations; see Section~\ref{sec:prelim}. Since $\oB[\ell'_{k-1,s},\ell'_{k-1,e}]$ contains only one symbol, $\ra$ queries that we need to answer are partial rank queries. Hence we can find the range of $\overline{T[i..i+j+k]}$ in time $O(1)$.  If $\oB[\ell'_{k-1,s},\ell'_{k-1,e}]$ contains more than one symbol, then  there is a node $u\in \cT$ that is labeled with $T[i..i+j+k-1]$; $u=lca(\ell_{k-1,s},\ell_{k-1,e})$ where $lca(f,g)$ denotes the lowest common ancestor of the $f$-th and the $g$-th leaves. We find the child $u'$ of the  node $u$ in $\cT$ that is labeled with $a=T[i+j+k]$. We also compute the Weiner link $\overline{u'}=\wlink(\overline{u},a)$ for a node $\overline{u'}=lca(\ell'_{k-1,s},\ell'_{k-1,e})$ in $\ocT$. Then $\ell'_{k,s}=\mathtt{leftmost\_leaf}(\overline{u'})$ and $\ell'_{k,e}=\mathtt{rightmost\_leaf}(\overline{u'})$. 
We need to visit at most $t$ nodes of $\cT$ and at most $t$ nodes of $\ocT$ in order to find the desired interval. 
By Proposition~\ref{prop:descendant} and Proposition~\ref{prop:path}, the total time needed to move down in $\cT$ is $O(t + \log\log \sigma)$. By Proposition~\ref{prop:wdescendant} and Proposition~\ref{prop:wpath}, the total time to compute all necessary w-links in $\ocT$ is also $O(t +\log\log \sigma)$. 
\end{proof}

\paragraph{Finding the Intervals.}
The algorithm for computing PLCP, described in Section~\ref{sec:permlcp}, assumes that we know the intervals  of $T[j\Delta'..j\Delta'+\ell_i]$ and 
$\overline{T[j\Delta'..j\Delta'+\ell_{i}]}$  for $i=j\Delta'$ and $j=0,1,\ldots, n/\Delta'$.  These values can be found as follows. We start by computing the intervals of $T[0..\ell_0]$ and $\overline{T[0..\ell_0]}$.  Suppose that the intervals of $T[j\Delta'..j\Delta'+\ell_i]$ and 
$\overline{T[j\Delta'..j\Delta'+\ell_{i}]}$ are known. We can compute $\ell_{(j+1)\Delta'}$ as shown in Section~\ref{sec:permlcp}.  We find the intervals of $T[(j+1)\Delta'..j\Delta'+\ell_{i}]$ and $\overline{T[(j+1)\Delta'..j\Delta'+\ell_{i}]}$ in time $O(\Delta')$ by executing $\Delta'$ operations $\contractleft$.   Each operation $\contractleft$ takes constant time. Then we calculate the intervals of $T[(j+1)\Delta'..(j+1)\Delta'+\ell_{i+1}]$ and $\overline{T[(j+1)\Delta'..(j+1)\Delta'+\ell_{i+1}]}$  in $O(\log\log\sigma+(\ell_{i+1}-\ell_i+\Delta'))$ time using Lemma~\ref{lemma:extendsequence}. We know from Section~\ref{sec:permlcp} that $\sum (\ell_{i+1}-\ell_i)=O(n)$. Hence we compute all necessary intervals in time $O(n+(n/\Delta')\log\log \sigma)=O(n)$.

\section{Compressed Index}
\label{sec:index}
In this section we show how our algorithms can be used to construct a compact index in deterministic linear time. We prove the following result. 
\begin{theorem}
  \label{theor:index}
We can construct an index for a text $T[0..n-1]$ over an alphabet of size $\sigma$ in $O(n)$ deterministic time using  O$(n\log\sigma)$ bits of working space. This index occupies $nH_k+ o(n\log\sigma)+O(n\frac{\log n}{d})$ bits of space for a parameter $d>0$. All occurrences of a query substring $P$ can be counted in $O(|P|+\log\log \sigma)$ time; all $\occ$ occurrences of $P$ can be reported in $O(|P|+\log\log\sigma + \occ\cdot d)$ time. An arbitrary substring $P$ of $T$  can be extracted in $O(|P|+d)$ time.
\end{theorem}

An uncompressed index by Fischer and Gawrychowski~\cite{FG13} also supports counting queries in $O(|P|+\log\log \sigma)$ time; however their data structure uses $\Theta(n\log n)$ bits. We refer to~\cite{BelazzouguiN14} for the latest results on compressed indexes. 
\paragraph{Interval Rank Queries.}
We start by showing how a compressed data structure that supports select queries can be extended to support a new kind of queries that we dub \emph{small interval rank queries}. An interval rank query $\ra_a(i,j,B)$ asks for $\ra_a(i',B)$ and $\ra_a(j',B)$, where $i'$ and $j'$ are the leftmost and rightmost occurrences of the  symbol $a$ in $B[i..j]$; if $a$ does not occur in $B[i..j]$, we return {\em null}. An interval query $\ra_a(i,j,B)$ is a small interval query if $j-i\le 2\log^2\sigma$.
Our compressed index relies on the following result.
\begin{lemma}
  \label{lemma:interrank}
Suppose that we are given a data structure that supports $\acc$ queries on a sequence $C[0..m]$ in time $t_{\sel}$. Then, using $O(m\log \log \sigma)$ additional bits, we can  support small  interval rank queries on $C$ in $O(t_{\sel})$ time.
\end{lemma}
\begin{proof}
 We split $C$ into groups $G_i$ so that every group contains $\log^2\sigma$ consecutive symbols of $S$, $G_i=C[i\log^2\sigma .. (i+1)\log^2\sigma-1]$. Let $A_i$ denote the set of symbols that occur in $G_i$. We would need  $\log \sigma$ bits per symbol to store $A_i$. Therefore we keep only a dictionary $A'_i$ implemented as a succinct  SB-tree~\cite{GrossiORR09}.  
A succinct SB-tree needs $O(\log\log m)$ bits per symbol; using SB-tree, we can determine whether a query symbol $a$ is in $A_i$ in constant time if we can access elements of $A_i$. We can identify every $a\in A_i$ by its leftmost position in $G_i$. Since $G_i$ consists of $\log^2\sigma$ consecutive symbols, a position within $G_i$ can be specified using $O(\log\log \sigma)$ bits. Hence we can  access any symbol of $A_i$ in $O(1)$ time. For each $a\in A_i$ we also keep a data structure $I_{a,i}$ that stores all positions where  $a$ occurs in $G_i$. Positions are stored as differences with the left border of $G_i$: if $C[j]=a$, we store the  difference $j-i\log^2\sigma$. Hence elements of $I_{a,i}$ can be stored in $O(\log\log\sigma)$ bits per symbol. $I_{a,i}$ is also implemented as an SB-tree. 

Using data structures $A'_i$ and $I_{a,i}$, we can answer small interval rank queries.  
Consider a group  $G_t=C[t\log^2\sigma..(t+1)\log^2\sigma-1]$, an index $i$ such that  $t\log^2\sigma \le i \le (t+1)\log^2\sigma$, and a symbol $a$. We can find the largest $j\le i$ 
such that $C[j]=a$ and $C[j]\in G_t$: first we look for the symbol $a$ in $A'_t$; if $a\in A'_t$, we find the predecessor of $j$ in $I_{a,t}$.  An interval $C[i..j]$ of size $d\le \log^2 \sigma$ intersects at most two groups $G_t$ and $G_{t-1}$.  We can find the rightmost occurrence of a symbol $a$ in $[i,j]$ as follows. First we look for the rightmost occurrence $j'\le j$ of $a$ in $G_t$; if $a$ does not occur in $C[t\log^2\sigma.. j]$, we look for the rightmost occurrence $j'\le t\log^2\sigma-1$ of $a$ in $G_{t-1}$. We can find the leftmost occurrence $i'$ of $a$ in $C[i..j]$ using a symmetric procedure. When $i'$ and $j'$ are found, we can compute $\ra_a(i',C)$ and $\ra_a(j',C)$ in $O(1)$ time by answering partial rank queries. Using the result of Theorem~\ref{theor:partrank} we can support partial rank queries in $O(1)$ time and $O(m\log\log\sigma)$ bits. 

Our data structure takes $O(m\log\log m)$ additional bits: Dictionaries $A'_i$ need $O(\log\log m)$ bits per symbol. Data structures $I_{a,t}$ and the structure for partial rank queries need $O(m\log\log \sigma)$ bits.  We can reduce the space usage from $O(m\log\log m)$ to $O(m\log\log \sigma)$ using the same method as in Theorem~\ref{theor:partrank}. 
\end{proof}

\paragraph{Compressed Index.}
We mark nodes of the suffix tree $\cT$ using the method of Section~\ref{sec:intervals}, but we set $d=\log \sigma$. 
Nodes of $\cT$ are classified into heavy, light, and special as defined in Section~\ref{sec:intervals}.  For every special node $u$, we construct a dictionary data structure $D_u$ that contains the labels of all heavy children of $u$. 
If there is child $u_j$ of $u$, such that the first symbol on the edge from to $u$ to $u_j$ is $a_j$, then we keep $a$ in $D_u$. For every $a_j\in D_u$ we store  the index $j$ of the child $u_j$.  If a heavy node $u$ has only one heavy child $u_j$ and more than $d$ light children, then we also store data structure $D_u$ for such a node $u$. If a heavy node has less $d$ children and one heavy child, then we keep the index of the heavy child using $O(\log d)=O(\log\log \sigma)$ bits. 

The second component of our index is the Burrows-Wheeler Transform $\oB$ of the reverse text $\oT$. We keep the data structure that supports partial rank, select, and access queries on $\oB$. Using e.g., the result from~\cite{BCGNNalgor13}, we can support $\acc$ queries in $O(1)$ time while $\ra$ and $\sel$ queries are answered in $O(\log\log \sigma)$ time. 
Moreover we construct a data structure, described in Lemma~\ref{lemma:interrank}, that supports rank queries on a small interval in $O(1)$ time. We also keep the data structure of Lemma~\ref{lemma:colrep} on $\oB$; using this data structure, we can find in $O(1)$ time whether an arbitrary interval  $\oB[l..r]$ contains exactly  one symbol. 
Finally we explicitly store answers to selected rank queries.  Let $\oB[l_u..r_u]$ denote the range of $\oP_u$, where $P_u$ is the string that corresponds to a node $u$ and $\oP_u$ is the reverse of $P_u$. For all data structures $D_u$ and for every symbol $a\in D_u$  we store the values of $\ra_a(l_u-1,\oB)$ and $\ra_a(r_u,\oB)$. 

We will show later in this section that each $\ra$ value can be stored in $O(\log\sigma)$ bits.  Thus $D_u$ needs $O(\log\sigma)$ bits per element. The total number of elements in all $D_u$ is  equal to the  number of special nodes plus the number of  heavy nodes with  one heavy child and at least $d$ light children. Hence all $D_u$ contain $O(n/d)$ symbols and use 
$O((n/d)\log\sigma)=O(n)$ bits of space. Indexes of heavy children for nodes with only one heavy child and at most $d$ light children can be kept in $O(\log\log\sigma)$ bits. 
 Data structure that supports $\sel$, $\ra$, and $\acc$  queries on $\oB$ uses $nH_k(T)+o(n\log\sigma)$ bits. Auxiliary data structures on $\oB$ need $O(n)+O(n\log\log\sigma)$ bits. Finally we need $O(n\frac{\log n}{d})$ bits to retrieve the position of a suffix in $\oT$ in $O(d)$ time. Hence the space usage of our data structure is $nH_k(T)+o(n\log\sigma)+O(n)+O(n\frac{\log n}{d})$.


\paragraph{Queries.}
Given a query  string $P$, we will find in time $O(|P|+\log\log \sigma)$  the range of the reversed string  $\oP$ in $\oB$. We will show below how to find the range of $\overline{P[0..i]}$ if the range of $\overline{P[0..i-1]}$ is known. Let $[l_j..r_j]$ denote the range of $\overline{P[0..j]}$, i.e., $\overline{P[0..j]}$ is the longest common prefix of all suffixes in $\oB[l_j..r_j]$.  We can compute $l_j$ and $r_j$ from $l_{j-1}$ and $r_{j-1}$ as $l_j=\Acc[a]+\ra_a(l_{j-1}-1,\oB)+1$ and $r_j=\Acc[a]+\ra_a(r_{j-1},\oB)$ for $a=P[j]$ and $j=0,\ldots, |P|$. Here  $\Acc[f]$ is the accumulated frequency of the 
first $f-1$ symbols. Using our auxiliary data structures on $\oB$ and additional information stored in nodes of the suffix tree $\cT$, we can answer necessary $\ra$ queries in constant time (with one exception).  At the same time we traverse a path in the suffix tree $\cT$ until the locus of $P$ is found or a light node is reached. Additional information stored in selected tree nodes will help us answer $\ra$ queries in constant time. A more detailed description is given below.

Our procedure starts at the root node of $\cT$ and we set $l_{-1}=0$, $r_{-1}=n-1$, and $i=0$. 
We compute the ranges $\oB[l_i..r_i]$ that correspond to $\overline{P[0..i]}$ for $i=0,\ldots, |P|$. Simultaneously we move down in  the suffix tree until we reach a light node. Let $u$ denote the last visited node of $\cT$ and let $a=P[i]$.  We denote by $u_a$ the next node that we must visit in the suffix tree, i.e., $u_a$ is the locus of $P[0..i]$. 
We can compute $l_i$ and $r_i$ in $O(1)$ time if $\ra_a(r_{i-1},\oB)$ and $\ra_a(l_{i-1}-1,\oB)$ are known.  We will show below that these queries can be answered in constant time because either (a)  the answers to $\ra$ queries are explicitly stored in $D_u$ or (b) the $\ra$ query that must be answered is a small interval $\ra$ query. The only exception is the situation when we move from a heavy node to a light node in the suffix tree; in this situation  the $\ra$ query takes $O(\log\log\sigma)$ time.  
For ease of description we distinguish between the following four cases. \\
({\bf i}) Node $u$ is a heavy node and $a\in D_u$. In this case we identify the heavy child $u_j$ of $u$ that is labeled with $a$. We can also find $l_i$ and $r_i$ in time $O(1)$ because $\ra_a(l_{i-1},\oB)$ and $\ra_a(r_{i-1},\oB)$ are stored  in $D_u$.\\
({\bf ii}) Node $u$ is a heavy node  and $a\not\in D_u$ or we do not keep the dictionary $D_u$ for the node $u$. In this case $u$ has at most one heavy child and at most $d$ light children.
If $u_a$ is a heavy node (case {\bf iia}), then the leftmost occurrence of $a$ in $\oB[l_{i-1}..r_{i-1}]$ is within $d^2$ symbols of $l_{i-1}$  and the rightmost occurrence of $a$ in $\oB[l_{i-1}..r_{i-1}]$ is within $d^2$ symbols of $r_{i-1}$. Hence we can find   $l_i$ and $r_i$  by answering small interval rank queries  $\ra_a(l_{i-1},l_{i-1}+d^2)$  and $\ra_a(r_{i-1}-d^2,r_{i-1})$  respectively. \\
If $u_a$ is a light node (case {\bf iib}), we answer two standard rank queries on $\oB$ in order to compute $l_i$ and $r_i$. \\  
({\bf iii}) If $u$ is a light node, then  $P[0..i-1]$ occurs at most $d$ times. Hence $\overline{P[0..i-1]}$ also occurs at most $d$ times and  $r_{i-1}-l_{i-1}\le d$. Therefore  we can compute $r_i$ and $l_i$ in $O(1)$ time by answering small interval rank queries. \\
({\bf iv}) We are on an edge of the suffix tree between a node $u$ and some child $u_j$ of $u$. 
In this case all occurrences of $P[0..i-1]$ are followed by the same symbol $a=P[i]$. Hence 
all occurrences of $\overline{P[0..i-1]}$ are preceded by $P[i]$ in the reverse text. Therefore $\oB[l_{i-1}..r_{i-1}]$ contains only one symbol $a=P[i]$. In this case $\ra_a(r_{i-1},\oB)$ and $\ra_a(l_{i-1}-1,\oB)$ are partial rank queries; hence $l_i$ and $r_i$  can be computed in $O(1)$ time. 

In all cases, except for the case (iia), we can answer $\ra$ queries and compute  $l_i$ and $r_i$ in $O(1)$ time. In case (iia) we need $O(\log\log \sigma)$ time answer $\ra$ queries. However case (iia) only takes place when the node $u$ is heavy and  its child $u_a$ is light. Since all descendants of a light node are light, case (iia) occurs only once when the pattern $P$ is processed.  Hence the total time to find the range of $\overline{P}$ in $\oB$ is $O(|P|+\log\log \sigma)$ time. When the range is known, we can count and report all occurrences of $\oP$ in standard way.

\paragraph{Construction Algorithm.}
We can construct the suffix tree $\cT$ and the BWT $\oB$ in $O(n)$ deterministic time. Then we can visit all nodes of $\cT$ and identify all nodes $u$ for which the data structure $D_u$ must be constructed. We keep information about nodes for which $D_u$ will be constructed in a bit vector. For every such node we also store the list of its heavy children with  their labels.  To compute additional information for $D_u$, we traverse the nodes of $\cT$ one more time using a variant of depth-first search. When a node $u\in \cT$ is reached, we know the interval $[l_u,r_u]$ of $\overline{s_u}$ in $\oB$, where $s_u$ is the string that labels the path from the root to a node $u\in \cT$.  We generate the list of all children $u_i$ of $u$ and their respective labels $a_i$.  If we store a data structure $D_u$ for the node $u$, we identify labels $a_h$ of heavy children $u_h$ of $u$. For every $a_h$ we compute $\ra_{a_h}(l_u-1,\oB)$
and $\ra_{a_h}(r_u,\oB)$ and add this information to $D_u$.  Then we generate the intervals that correspond to all strings $\overline{s_ua_i}$ in $\oB$ and keep them in a list $List(u)$. Since intervals in $List(u)$ are disjoint, we can store $List(u)$  in $O(\sigma \log n)$ bits.  

We can organize our  traversal in such way that only $O(\log n)$ lists $List(u)$ need to be stored. Let $num(u)$ denote the number of leaves in the subtree of a node $u$. We say that a node is \emph{small} if $num(u_i)\le num(u)/2$ and big otherwise. Every node can have at most one big child. When a node $u$  processed and $List(u)$ is generated, we visit small children $u_i$ of $u$ in arbitrary order. When all small children $u_i$ are visited and processed, we discard the list $L(u)$. Finally if $u$ has a big child $u_b$, we visit $u_b$. If a node $u$ is not the root node and we keep $List(u)$, then $num(u)\le num(parent(u))/2$. Therefore we keep $List(u)$ for at most $O(\log n)$ nodes $u$. Thus the  space we need to store  all $List(u)$ is $O(\sigma\log^2 n)=o(n)$ for $\sigma\le n^{1/2}$.  Hence the total workspace used of our algorithm is $O(n\log\sigma)$.
The total number of $\ra$ queries that we need to answer is $O(n/d)$ because all $D_u$ contain $O(n/d)$ elements. We need $O((n/d)\log\log\sigma)$ time to construct all $D_u$ and to answer all $\ra$ queries.  The total time needed to traverse $\cT$ and collect necessary data about heavy nodes and special nodes is $O(n)$. Therefore our index can be constructed in $O(n)$ time. 

It remains to show how we can store selected precomputed answers to $\ra$ queries in $O(\log\sigma)$ bits per query. We divide the sequence $\oB$ into chunks of size $\sigma^2$. For each chunk and for every symbol $a$ we encode the number of $a$'s occurrences per chunk in a binary sequence $A_a$, $A_a=1^{d_1}01^{d_2}0\ldots1^{d_i}0\ldots$ where $d_i$ is equal to the  number of times $a$ occurs in the $i$-th chunk. If a symbol $\oB[i]$ is in the chunk $Ch$, then we can answer $\ra_a(i,\oB)$ by $O(1)$  queries on $A_a$ and a rank query on $Ch$; see e.g.,~\cite{GolynskiMR06}.   Suppose that we need to store a pre-computed answer to a query $\ra_a(i,\oB)$; we store the answer to $\ra_a(i',Ch)$ where $Ch$ is the chunk that contains $i$ and $i'$ is the relative position of $\oB[i]$ in $Ch$.  Since a chunk contain $\sigma^2$ symbols, $\ra_a(i',Ch)\le \sigma^2$ and we can store the answer to $\ra_a(i',Ch)$ in $O(\log \sigma)$ bits. When the answer to the rank query on $Ch$ is known, we can compute the answer to $\ra_a(i,\oB)$ in $O(1)$ time.

\end{document}